\documentclass[11pt]{article}

\usepackage{fullpage}

\usepackage{comment}
\usepackage{epsf}
\usepackage{graphicx}
\usepackage{subfigure}
\usepackage{bbm}
\usepackage{color}
\usepackage{nicefrac}
\usepackage{mathtools}
\usepackage{amsfonts}
\usepackage{amsthm}
\usepackage{amsmath, bm}
\usepackage{amssymb}
\usepackage{textcomp}
\usepackage{xcolor}
\usepackage[ruled, vlined, lined, commentsnumbered]{algorithm2e}

\usepackage{hyperref}
\definecolor{OliveGreen}{HTML}{3C8031}
\hypersetup{
colorlinks=true,
urlcolor=blue,
linkcolor=blue,
citecolor=OliveGreen,
}

\newenvironment{fminipage}%
  {\end{minipage}\end{Sbox}\fbox{\TheSbox}}

\newtheorem{theorem}{Theorem}[section]

\newtheorem{proposition}[theorem]{Proposition}
\newtheorem{lemma}[theorem]{Lemma}

\newtheorem{definition}[theorem]{Definition}

\newtheorem{claim}[theorem]{Claim}

\newlength{\widebarargwidth}
\newlength{\widebarargheight}
\newlength{\widebarargdepth}

\makeatletter
\long\def\@makecaption#1#2{
       \vskip 0.8ex
       \setbox\@tempboxa\hbox{\small {\bf #1:} #2}
       \parindent 1.5em 
       \dimen0=\hsize
       \advance\dimen0 by -3em
       \ifdim \wd\@tempboxa >\dimen0
               \hbox to \hsize{
                       \parindent 0em
                       \hfil 
                       \parbox{\dimen0}{\def\baselinestretch{0.96}\small
                               {\bf #1.} #2
                               } 
                       \hfil}
       \else \hbox to \hsize{\hfil \box\@tempboxa \hfil}
       \fi
       }
\makeatother

\long\def\comment#1{}

\newcommand{\defn}{\coloneqq}

\DeclareMathOperator{\Var}{Var}
\DeclareMathOperator{\Cov}{Cov}

\newcommand{\EE}{\ensuremath{{\mathbb{E}}}}

\newcommand{\BER}{\ensuremath{\mbox{\sf Ber}}}
\newcommand{\BIN}{\ensuremath{\mbox{\sf Bin}}}

\newcommand{\calE}{\mathcal{E}}
\newcommand{\calP}{\mathcal{P}}
\newcommand{\calQ}{\mathcal{Q}}
\newcommand{\calS}{\mathcal{S}}
\newcommand{\calC}{\mathcal{C}}
\newcommand{\bbone}[1]{\mathbbm{1}\{#1\}}
\newcommand{\set}[1]{\{#1\}}
\newcommand{\defeq}{\coloneqq}
\newcommand{\N}{\mathbb{N}}

\newcommand{\R}{\mathbb{R}}
\newcommand{\E}{\mathbb{E}}
\renewcommand{\P}{\mathbb{P}}

\title{Detection of Dense Subhypergraphs by Low-Degree Polynomials}

\usepackage{authblk}

\author[1]{Abhishek Dhawan\thanks{Email: \textit{abhishekdhawan@gatech.edu}. A.D.\ was supported in part by NSF grant DMS-2053333.}}

\author[2]{Cheng Mao\thanks{Email: \textit{cheng.mao@math.gatech.edu}. C.M.\ was supported in part by NSF grants DMS-2053333 and DMS-2210734.}}

\author[3]{Alexander S.\ Wein\thanks{Email: \textit{aswein@ucdavis.edu}. Part of this work was done while A.S.W.\ was with the Algorithms and Randomness Center at Georgia Tech, supported by NSF grants CCF-2007443 and CCF-2106444.}}

\affil[1,2]{School of Mathematics, Georgia Institute of Technology}

\affil[3]{Department of Mathematics,
University of California, Davis}

\date{}

\begin{document}

\maketitle

\begin{abstract}
Detection of a planted dense subgraph in a random graph is a fundamental statistical and computational problem that has been extensively studied in recent years. We study a hypergraph version of the problem. Let $G^r(n,p)$ denote the $r$-uniform Erd\H{o}s-R\'enyi hypergraph model with $n$ vertices and edge density $p$. We consider detecting the presence of a planted $G^r(n^\gamma, n^{-\alpha})$ subhypergraph in a $G^r(n, n^{-\beta})$ hypergraph, where $0< \alpha < \beta < r-1$ and $0 < \gamma < 1$. Focusing on tests that are degree-$n^{o(1)}$ polynomials of the entries of the adjacency tensor, we determine the threshold between the easy and hard regimes for the detection problem. More precisely, for $0 < \gamma < 1/2$, the threshold is given by $\alpha = \beta \gamma$, and for $1/2 \le \gamma < 1$, the threshold is given by $\alpha = \beta/2 + r(\gamma - 1/2)$.

Our results are already new in the graph case $r=2$, as we consider the subtle \emph{log-density regime} where hardness based on average-case reductions is not known. Our proof of low-degree hardness is based on a \emph{conditional} variant of the standard low-degree likelihood calculation.
\end{abstract}

\tableofcontents

\section{Introduction}\label{section:intro}

Finding a dense subgraph in a given undirected graph is an iconic problem at the intersection of graph theory, computer science, and statistics, and it finds broad applications in social and biological sciences.
The last few decades have observed a wide range of research on multiple variants of the problem, including the planted clique problem \cite{karp1976probabilistic,jerrum1992large}, densest $k$-subgraph problem \cite{feige1997densest,log-density}, community detection \cite{fortunato2010community,AV-detection}, and hypergraph versions \cite{chlamtac2018densest,corinzia2022statistical,luo2022tensor}, among others.
In this work, we focus on the problem of detecting the presence of a \emph{planted dense subhypergraph} in a given $r$-uniform hypergraph $H$ on $n$ vertices.
In the language of statistical hypothesis testing, we consider the task of testing between the following two distributions on $H$ (defined formally in Section~\ref{section:main_results}):
\begin{itemize}
\item the \emph{null distribution} is the $r$-uniform Erd\H{o}s-R\'enyi hypergraph $G^r(n,q)$ with edge density $q$;
\item the \emph{planted distribution} randomly selects $\approx \rho n$ vertices to lie in the planted subhypergraph, the edges within the subhypergraph occur independently with probability $p$, and all other edges occur independently with probability $q$.
\end{itemize}
Given a single hypergraph $H$ drawn from one of these two distributions, the goal is to distinguish the two cases with high probability (w.h.p.), i.e., probability $1-o(1)$ as $n \to \infty$, where $p,q,\rho$ may scale with $n$. We assume the parameters $p,q,\rho$ are known.

Our focus is on understanding the power and limitation of \emph{computationally efficient} tests, i.e., polynomial-time algorithms. We currently lack complexity-theoretic tools to prove computational hardness of \emph{average-case} problems like this one (where the input is random), so the leading approaches are either based on average-case reductions which formally relate different average-case problems to each other (see e.g.~\cite{BB-secret} and references therein) or unconditional lower bounds against restricted classes of algorithms.

Our main result establishes sharp necessary and sufficient conditions on $p,q,\rho$ for success of \emph{low-degree polynomial tests}. This is a powerful class of tests including statistics like small subgraph counts (edges, triangles, etc.). It is by now well-established that these low-degree tests are a useful proxy for computationally efficient tests, in the sense that the best-known polynomial-time algorithms for a wide variety of high-dimensional testing problems are captured by the low-degree class; see e.g.~\cite{sam-thesis,kunisky2019notes}.

More specifically, we focus on the so-called \emph{log-density} regime \cite{log-density} where $p=n^{-\alpha}$, $q=n^{-\beta}$, and $\rho = n^{\gamma-1}$ for constants $0 < \alpha < \beta < r-1$ and $\gamma \in (0,1)$. The assumption $\beta < r-1$ precludes isolated vertices. We show that if $0 < \gamma < 1/2$ and $\alpha < \beta \gamma$, or if $1/2 \le \gamma < 1$ and $\alpha < \beta/2 + r(\gamma - 1/2)$, there is a constant-degree polynomial test that distinguishes the null and the planted distribution w.h.p. 
On the other hand, if $0 < \gamma < 1/2$ and $\alpha > \beta \gamma$, or if $1/2 \le \gamma < 1$ and $\alpha > \beta/2 + r(\gamma - 1/2)$, then no degree-$n^{o(1)}$ polynomial can separate the null and the planted distribution.

\subsection{Relation to prior work}

\paragraph{Planted dense subgraph.}

Our results are already interesting in the graph case $r = 2$, in which case we are considering the \emph{planted dense subgraph detection} problem. The statistical threshold is well-studied~\cite{AV-detection,BI-detection,VA-detection-sparse,HWX-reduction}, and the computational threshold has also been established in various parameter regimes via reduction from the presumed-hard \emph{planted clique} problem~\cite{HWX-reduction,BBH-reduction,BBH-univ}. However, existing reductions assume that $p$ and $q$ are of the same order, and so we are currently lacking reduction-based evidence for hardness in the log-density regime where $p=n^{-\alpha}$, $q=n^{-\beta}$, and $\rho = n^{\gamma-1}$ for constants $0 < \alpha < \beta < 1$ and $\gamma \in (0,1)$. Here the problem undergoes a qualitative change at $\gamma = 1/2$. Namely, when $\gamma \ge 1/2$, the best-known polynomial-time algorithm is simply to threshold the total number of edges in the graph, and this succeeds when $\alpha < \beta/2 + 2(\gamma-1/2)$. On the other hand, when $\gamma < 1/2$, the best-known polynomial-time algorithm is a more subtle subgraph-counting procedure which succeeds when
$\alpha < \beta \gamma$; see Section~3.2 of~\cite{log-density} (they give only a proof sketch, as their main goal is to give guarantees for approximating the densest $k$-vertex subgraph in \emph{worst-case} graphs).

Our main result shows that for every $\gamma \in (0,1)$, low-degree tests cannot surpass the thresholds above, providing evidence for optimality of the existing algorithms. We also give a matching low-degree testing upper bound, confirming that the above algorithms are captured by the low-degree framework; in the case $\gamma < 1/2$, our test is simpler than that of~\cite{log-density}, although their algorithm can also \emph{recover} the planted subgraph.

In contrast, the \emph{statistical} (i.e., information-theoretic) threshold for detection (i.e., testing) is distinct from the computational threshold above, at least in some parameter regimes~\cite{AV-detection,VA-detection-sparse}.
Specifically, there is a gap between the statistical and computational thresholds whenever $\gamma \in (0,1/2]$. Also, when $\gamma \in (1/2,2/3)$ there is a gap for some values of $\beta$. We will elaborate further in our discussion of the hypergraph case below.

While it is not our main focus, we note that the statistical and computational aspects of \emph{recovering} a planted dense subgraph (i.e., identifying the planted vertices)\ are also well-studied~\cite{ames-recovery,CX-growing,M-one-comm,HWX-info}. Formally, the recovery problem is at least as hard as the detection problem (via a polynomial-time reduction in the style of~\cite[Section~5.1]{MW-reduction}). Notably, when $\gamma > 1/2$, the recovery problem is \emph{strictly} harder, i.e., the computational thresholds for detection and recovery are different (at least for low-degree algorithms), with recovery requiring $\alpha < \beta/2+\gamma-1/2$~\cite{SW-estimation}. When $\gamma < 1/2$, our result shows low-degree hardness of detection at the same threshold as the recovery algorithm of~\cite{log-density}, so there is no \emph{detection-recovery gap} in this regime, resolving a question left open by \cite[Section~2.4.1]{SW-estimation}.

\paragraph{Low-degree testing.}

A successful degree-$D$ test is a degree-$D$ multivariate polynomial in the input variables (in our case, the $\binom{n}{r}$ hyperedge-indicator variables) whose real-valued output \emph{separates} (see Definition~\ref{def:strong-separation}) samples from the planted and null distributions. The idea to study this class of tests emerged from the line of work~\cite{sos-clique,HS-bayesian,sos-hidden,sam-thesis}; see also~\cite{kunisky2019notes} for a survey. Tests of degree $O(\log n)$ are generally taken as a proxy for polynomial-time tests, as they capture leading algorithmic approaches such as spectral methods. Our upper bound will give a constant-degree test, which yields a polynomial-time algorithm for testing; our lower bound will rule out all tests of degree $n^{o(1)}$.

There is by now a standard method for proving low-degree testing bounds based on the \emph{low-degree likelihood ratio} (see Section~2.3 of~\cite{sam-thesis}), which boils down to finding an orthonormal basis of polynomials with respect to the null distribution, and computing the expectations of these basis polynomials under the planted distribution. However, our setting is more subtle because (for $\gamma < 1/2$) the second moment of the low-degree likelihood ratio diverges due to rare ``bad'' events under the planted distribution. We therefore need to carry out a conditional low-degree argument whereby the planted distribution is conditioned on some ``good'' event.

Conditional low-degree arguments of this kind have appeared before in a few instances~\cite{fp,grp-testing}, but our argument differs on a technical level. Prior work chose to condition on an event that would seem to make direct computations with the orthogonal polynomials very complicated; to overcome this, they bound the conditional low-degree likelihood ratio in an indirect way by first relating it to a certain ``low-overlap'' second moment (also called the \emph{Franz-Parisi criterion} in~\cite{fp}). In contrast our approach is more direct: we are careful to condition on an event for which computations with the orthogonal polynomials remain tractable.

\paragraph{Integrality gaps for densest $k$-subgraph.}

Complementary to our results is a line of work on establishing integrality gaps for convex relaxations of densest $k$-subgraph~\cite{BCGVZ-integrality,chlamtavc2018sherali,sos-dense-subgraph} (and~\cite{chlamtavc2018sherali} also considers the extension to hypergraphs). Here the setting of interest is to find, in a \emph{worst-case} graph, a subgraph induced on $k$ vertices (for a given $k$) whose density of edges is guaranteed to be within some factor of the densest such $k$-subgraph. The best-known polynomial-time approximation factor is $\approx n^{1/4}$, due to~\cite{log-density}. The integrality gap results mentioned above show that powerful families of convex relaxations fail to improve upon this.

In fact, the proofs of these integrality gaps show that the convex relaxations fail on certain random distributions such as the ones we study. The hard instance for approximation factor $n^{1/4-\epsilon}$ is precisely our testing problem with parameters $\alpha = 1/4+\epsilon$, $\beta = 1/2$, $\gamma = 1/2$ (and so our result provides additional evidence that this approximation factor is unimprovable). The work on convex relaxations that is most relevant to our result is the recent result of~\cite{sos-dense-subgraph} which was obtained independently from ours. We provide a detailed comparison below.

\paragraph{Comparison to independent work on sum-of-squares~\cite{sos-dense-subgraph}.}

In concurrent and independent work,~\cite{sos-dense-subgraph} proved that the \emph{sum-of-squares (SoS) hierarchy} (a powerful family of semidefinite programming relaxations) at degree $n^{\Omega(1)}$ fails to improve the approximation factor $n^{1/4}$ discussed above. This is the strongest known result on integrality gaps for the densest $k$-subgraph problem. While they are conjectured to be closely related (see~\cite{sos-hidden,sam-thesis}), we are not aware of any formal implications in either direction between failure of SoS and failure of low-degree tests.

More specifically, the authors of~\cite{sos-dense-subgraph} consider the same hypothesis testing problem that we do (in the graph case $r=2$) and write down a particular SoS program that takes as input a sample from the null distribution $G(n,q)$ and attempts to \emph{refute} the existence of a $(1+o(1))\rho n$-vertex subgraph with $(1+o(1)) p (\rho n)^2/2$ edges. Note that a subgraph of this density exists in the planted distribution with high probability, and thus, if SoS succeeds at the refutation task (w.h.p.)\ this also gives an algorithm for the detection (i.e., testing) task with the same parameters $p,q,\rho$. The result of~\cite{sos-dense-subgraph} shows that SoS \emph{fails} to solve the refutation task when $\gamma < 1/2$ and $\alpha > \beta \gamma$, matching the threshold in our result. This gives strong evidence for computational hardness of the \emph{refutation} problem in this regime. However, this need not imply hardness of the \emph{detection} problem (the focus of our work), since detection is a formally easier problem.

To illustrate the previous point, when $\gamma > 1/2$ (a regime not covered by~\cite{sos-dense-subgraph}) we expect an inherent gap between the detection and refutation tasks, with polynomial-time detection requiring $\alpha < \beta/2 + 2(\gamma-1/2)$ and polynomial-time refutation requiring $\alpha < \beta/2 + \gamma - 1/2$; see Appendix~\ref{app:det-ref} for more details. Therefore, in this regime, we expect that checking feasibility of the SoS program of~\cite{sos-dense-subgraph} is a strictly suboptimal algorithm for the \emph{detection} problem. Accordingly, one strength of our result is that it directly addresses the detection problem and captures the best-known computational threshold for all $\gamma \in (0,1)$.

The above discussion warns that one should be careful when taking an SoS lower bound as evidence for hardness of \emph{detection} (rather than refutation). However, in this case the situation is more subtle because the proof of~\cite{sos-dense-subgraph} uses the \emph{pseudo-calibration} approach~\cite{sos-clique}, and so a key step in their analysis is closely related to the low-degree testing lower bounds that we prove. Namely they prove that the low-degree likelihood ratio $L_{\le D}$ (in their notation, $\widetilde{\EE}[1]$) is $1+o(1)$ with high probability over the null distribution (and like our result, this involves some conditioning arguments). Our approach is related but distinct: we show that the second moment of the \emph{conditional} low-degree likelihood ratio (see Section~\ref{section:lower_bound}) is $1+o(1)$, and this is what implies our desired result (failure of weak separation). The proof of~\cite{sos-dense-subgraph} also requires some additional steps, notably a rather involved and impressive analysis to show that their moment matrix is positive semidefinite.

\paragraph{Planted dense subhypergraphs.}
Compared to planted dense subgraphs, the generalized problem for hypergraphs is less well-studied.
For the worst-case densest subhypergraph problem, \cite{applebaum2012pseudorandom} shows that it is hard to approximate the densest subhypergraph to within an $n^\epsilon$ factor for a fixed $\epsilon > 0$ based on a pseudorandomness assumption.
In the line of research on integrality gaps discussed above, the result of \cite{chlamtavc2018sherali} holds for $r$-uniform hypergraphs, establishing an integrality gap of $\approx n^{(r-1)/4}$ for the Sherali-Adams hierarchy, which matches the log-density threshold.
However, unlike the graph case $r=2$, no polynomial-time approximation algorithm is known to match the above threshold for $r \ge 4$, and partial progress has been made for $r = 3$ by \cite{chlamtac2018densest}.

Considering the same planted dense subhypergraph model as ours, \cite{luo2022tensor} studies both the detection and the recovery problem.
For the detection problem with $p = 1$ and $q = 1/2$ (i.e., planted clique), it is proved that degree-$O(\log n)$ polynomial tests fail if $\gamma < 1/2$.
This is complementary to our result as we consider $p = n^{-\alpha}$ and $q = n^{-\beta}$ with $0 < \alpha < \beta < r-1$.
For the recovery problem, \cite{luo2022tensor} focuses on the regime $\gamma \ge 1/2$ and proposes an efficient algorithm succeeding when $\alpha < \beta/2 + (r-1)(\gamma - 1/2)$.
In addition, if $\alpha > \beta/2 + (r-1)(\gamma - 1/2)$, it is shown that polynomials of degree $\operatorname{polylog}(n)$ fail to recover the planted dense subhypergraph.
Note that this recovery threshold is different from the detection threshold $\alpha = \beta/2 + r (\gamma - 1/2)$ proved in this work; therefore, these two results together generalize the detection-recovery gap in the regime $\gamma \ge 1/2$ shown in \cite{SW-estimation} for $r = 2$ to the case of hypergraphs.

Another line of research
\cite{yuan2022sharp,yuan2021heterogeneous,yuan2021information} concerns the statistical thresholds for planted dense subhypergraphs.
In particular, in the same vein as the results from \cite{AV-detection}, it is shown in \cite{yuan2022sharp} that the optimal statistical threshold for detection is achieved either by a total degree test or a computationally inefficient scan test.
When the former dominates, the detection threshold naturally matches our result, but when the scan test performs better, there is a statistical-computational gap.
To be more precise, the total degree test succeeds when $\alpha < \alpha_{\mathrm{deg}} \defeq \beta/2 + r(\gamma-1/2)$ and the scan test succeeds when
$\alpha < \alpha_{\mathrm{scan}} \defeq \gamma (r-1)$.
If $\gamma \le 1/2$, the scan test strictly outperforms our condition $\alpha < \beta \gamma$, so there is a statistical-computational gap. Moreover, we also have a gap when $1/2 < \gamma < r/(r+1)$ and $2r(\gamma-1/2) < \beta < r-2\gamma$ (where the first condition on $\gamma$ ensures the second interval for $\beta$ is nonempty), since in this case $\alpha_{\mathrm{scan}} > \alpha_{\mathrm{deg}}$ and $\alpha_{\mathrm{deg}} < \beta$ (note that when $\alpha_{\mathrm{deg}} \ge \beta$, the total degree test succeeds for all $\alpha < \beta$ and there is no gap).

In addition, \cite{corinzia2019exact,corinzia2022statistical} study the recovery of a planted dense subhypergraph via a tensor PCA model with additive Gaussian noise.
They probe the computational threshold of the problem using Approximate Message Passing, but the result is not directly comparable to ours.
Finally, we refer the reader to a recent survey \cite[Section~5.7]{lanciano2023survey} for more related works on dense subhypergraphs.

\subsection*{Notation}
Let $\N$ denote the set of postive integers.
For any $n \in \N$, let $[n] := \{1,2,\dots,n\}$.
Throughout this work, we consider $n \to \infty$ and use the asymptotic notation $O(\cdot)$, $o(\cdot)$, etc.

For a fixed integer $r \ge 2$, let $K_n^r$ denote the complete $r$-uniform hypergraph on $n$ vertices.
For any hypergraph $H$, let $V(H)$ denote its vertex set and let $E(H)$ denote its edge set.
For a hypergraph $H$ and $S \subseteq E(H)$, let $H[S]$ denote the subgraph of $H$ induced by $S$.
Note that any set of hyperedges $S \subseteq E(K_n^r)$ can be identified with the subhypergraph $K_n^r[S]$, so we sometimes view $S$ as a subhypergraph without ambiguity.
For brevity, we often refer to hypergraphs and hyperedges simply as graphs and edges, respectively.

For a distribution $\calP$, let $\E_\calP$ denote the expectation under $\calP$, and with slight abuse of notation, let $\calP$ also denote the associated probability.
Let $\Var_{\calP}(\cdot)$ denote the variance of a random variable under $\calP$.
Let $\BER(p)$ denote the Bernoulli distribution with parameter $p \in [0,1]$, and let $\BIN(n,p)$ denote the Binomial distribution with parameters $n \in \N$ and $p \in [0,1]$.

\section{Main results}\label{section:main_results}

Let $n, r \in \N$ with $n \ge r \ge 2$, and let $ 0 < \alpha < \beta < r-1$ and $\gamma \in (0,1)$. Define $M \defeq \binom{n}{r}$.
We will be considering random $n$-vertex $r$-uniform hypergraph models where the output is an undirected hypergraph $H$ with adjacency tensor $Y \in \{0, 1\}^M$.
To be more precise, a hyperedge $e$ is a subset of $[n]$ of cardinality $r$; we let $[M] \defeq \binom{[n]}{r}$ denote the set of all hyperedges of $K_n^r$ and write $e \in [M]$.
The adjacency tensor $Y$ is indexed by $e \in [M]$, and $Y_e \in \{0,1\}$ indicates the presence of hyperedge $e$.

We formulate the detection of a dense random subhypergraph as a statistical hypothesis testing problem between distributions $\calP$ and $\calQ$, defined as follows:
\begin{itemize}
\item
Under $\calP$, let $z_1, \ldots, z_n$ be i.i.d.\ random variables from the Bernoulli distribution $\BER(\rho)$, where $\rho \defeq n^{\gamma - 1}$, and let $Z \defeq \{i \in [n] \,:\, z_i = 1\}$.
Conditional on $Z$, we observe an $r$-uniform hypergraph $H$ with independent hyperedges
\[Y_{e} \sim \left\{\begin{array}{cc}
    \BER(q) & \text{if } e \not \subseteq Z, \\
    \BER(p) & \text{if } e \subseteq Z,
\end{array}\right.\]
where $p \defeq n^{-\alpha}$ and $q \defeq n^{-\beta}$.

\item 
Under $\calQ$, we observe an $r$-uniform Erd\H{o}s-R\'enyi hypergraph $H$ with adjacency tensor $Y$, where the hyperedges are i.i.d.\ $Y_e \sim \BER(q)$ with $q \defeq n^{-\beta}$. 
\end{itemize}

To probe the computational threshold for testing between $\calP$ and $\calQ$, we focus on low-degree polynomial algorithms (e.g.,~\cite{sam-thesis,kunisky2019notes}).
Let $\R[Y]_{\le D}$ denote the set of multivariate polynomials in the entries of $Y$ with degree at most $D$.
With some abuse of notation, we will often say ``a polynomial'' to mean a sequence of polynomials $f = f_n \in \R[Y]_{\le D}$, one for each problem size $n$; the degree $D = D_n$ of such a polynomial may scale with $n$. To study the ability of a polynomial in testing $\calP$ against $\calQ$, we consider the notions of \emph{strong separation} and \emph{weak separation} defined in \cite{fp}, with the former being stronger than the latter.

\begin{definition}[\cite{fp}, Definition~1.6]
\label{def:strong-separation}
As $n \to \infty$, a polynomial $f \in \R[Y]_{\le D}$ is said to
\begin{itemize}
\item
strongly separate $\calP$ and $\calQ$ if
\( \sqrt{\Var_{\calP}(f(Y)) \vee \Var_{\calQ}(f(Y))} = o\left(\left|\E_{\calP}[f(Y)] - \E_{\calQ}[f(Y)]\right|\right) \);

\item
weakly separate $\calP$ and $\calQ$ if
\( \sqrt{\Var_{\calP}(f(Y)) \vee \Var_{\calQ}(f(Y))} = O\left(\left|\E_{\calP}[f(Y)] - \E_{\calQ}[f(Y)]\right|\right) \).
\end{itemize}
\end{definition}
\noindent
See \cite{fp} for a detailed discussion on why these conditions are natural for hypothesis testing.
In particular, by Chebyshev's inequality, strong separation implies that we can threshold $f(Y)$ to
test $\calP$ against $\calQ$ with vanishing type I and type II errors.
Our main results are the following.

\begin{theorem} 
\label{thm:main-results}
Suppose that we observe a random $n$-vertex $r$-uniform hypergraph $Y \in \{0,1\}^M$ from either $\calP$ or $\calQ$ with parameters $p = n^{-\alpha}$, $q = n^{-\beta}$, and $\rho = n^{\gamma - 1}$ for fixed $r \ge 2$, $0 < \alpha < \beta < r-1$, $\gamma \in (0,1)$, $n \to \infty$, and $M = \binom{n}{r}$.
\begin{itemize}
\item
(Lower bound)
Suppose that either (1) $\gamma \ge 1/2$ and $\alpha > \beta/2 + r(\gamma - 1/2)$, or (2) $\gamma < 1/2$ and $\alpha > \beta \gamma$.
If $D = n^{o(1)}$ then no polynomial in $\R[Y]_{\le D}$ weakly separates $\calP$ and $\calQ$.

\item
(Upper bound)
Suppose that either (1) $\gamma \ge 1/2$ and $\alpha < \beta/2 + r(\gamma - 1/2)$, or (2) $\gamma < 1/2$ and $\alpha < \beta \gamma$.
There exists a positive integer $D$ depending only on $(\alpha, \beta, \gamma)$ and a polynomial in $\R[Y]_{\le D}$ that strongly separates $\calP$ and $\calQ$.
\end{itemize}
\end{theorem}

\noindent
We have therefore completely characterized the low-degree detection threshold for both $\gamma \ge 1/2$ and $\gamma < 1/2$.
The lower bound in Theorem~\ref{thm:main-results} is proved at the beginning of Section~\ref{section:lower_bound}.
The upper bound is proved at the beginning of Section~\ref{section:upper_bound}.
In particular, in each regime stated in the upper bound, the strong separation is achieved by a constant-degree polynomial in the entries of $Y$, so the testing algorithm is polynomial-time.

\section{Detection lower bound}\label{section:lower_bound}

Following the framework of low-degree polynomial algorithms \cite{sam-thesis,kunisky2019notes,fp}, we first introduce the notation used in this section. 
Let $\tilde Y_e$ be the standardized hyperedge under $\calQ$, i.e., 
\begin{equation}
\Tilde{Y}_e \defeq \frac{Y_e - q}{\sigma},
\label{eq:def-tilde-y-e}
\end{equation}
where $\sigma \defeq \sqrt{\Var_{\calQ}(Y_e)} = \sqrt{q(1-q)}$.
We define
\begin{equation}
    \phi_S(Y) \defeq \prod_{e \in S}\Tilde{Y}_e, \quad \forall S \subseteq [M], \label{eqn:phi_defn}
\end{equation}
Note that $Y_e$ and $Y_f$ are independent under $\calQ$ for $e\neq f$.
Therefore, as $\tilde Y_e$ is centered, we have
\[\EE_\calQ[\phi_S(Y)\phi_{S'}(Y)] = \left\{\begin{array}{cc}
    0 & \text{ if } S \neq S', \\
    1 & \text{ if } S = S',
\end{array}\right.\]
so that $\{\phi_S : S \subseteq [M]\}$ is an orthonormal basis of the set of functions of $Y$.

In order to study polynomials of degree at most $D > 0$, we define following quantity
\begin{equation}
\label{eq:low-deg-likelihood-ratio-norm-def}
\|L_{\leq D}\|^2 \defeq \sum_{\substack{S \subseteq [M], \\ |S| \leq D}} \left( \EE_{\calP}[\phi_S(Y)] \right)^2 .
\end{equation}
When considering an event $\calE$ and the conditional distribution $\calP'$ on $\calE$, we define analogously
\begin{equation}
\|L_{\leq D}'\|^2 \defeq \sum_{\substack{S \subseteq [M], \\ |S| \leq D}}(\EE_{\calP'}[\phi_S(Y)])^2.
\label{eq:conditional-degree-d-norm}
\end{equation}
The notation $\|L_{\le D}\|$ stands for the norm of the degree-$D$ likelihood ratio between $\calP$ and $\calQ$.
We remark that $\|L_{\le D}\|$ is the same as the quantity LD$(D)$ in \cite[Definition~1.3]{fp}, and the equivalence is justified in \cite[Section~2.3]{sam-thesis}.

\begin{proof}[Proof of Theorem~\ref{thm:main-results} (lower bound)]
Let $\calE$ be an event such that $\calP(\calE) = 1 - o(1)$.
Let $\calP'$ be the distribution obtained from $\calP$ by conditioning on the event $\calE$.
By \cite[Proposition~6.2]{fp}, to prove that no polynomial in $\R[Y]_{\le D}$ weakly separates $\calP$ and $\calQ$, it suffices to show that 
\begin{equation}
\|L'_{\le D}\|^2 = 1 + o(1) .
\label{eq:ldlr-1-o-1}
\end{equation}
If $\gamma \ge 1/2$ and $\alpha > \beta/2 + r(\gamma - 1/2)$, we take $\calE$ to be the full sample space so that $\calP = \calP'$ and establish \eqref{eq:ldlr-1-o-1} in Proposition~\ref{prop:lower-large-gamma}.
If $\gamma < 1/2$ and $\alpha > \beta \gamma$, we define the high-probability event $\calE$ in \eqref{eq:def-cale} and establish \eqref{eq:ldlr-1-o-1} in Proposition~\ref{prop:lower-small-gamma}.
\end{proof}

\subsection{Lower bound for large $\gamma$}\label{subsection:large_gamma_lb}

We will first consider the case that $\gamma \geq 1/2$, i.e., the case that the set $Z$ is expected to be ``large''.

\begin{proposition}
\label{prop:lower-large-gamma}
Suppose $p = n^{-\alpha}$, $q = n^{-\beta}$, and $\rho = n^{\gamma - 1}$ with fixed $r \ge 2$, $0 < \alpha < \beta < r-1$, and $\gamma \in (0,1)$ such that
\begin{align}\label{eqn:large_gamma_lb}
    \gamma \geq \frac{1}{2}, \quad \alpha > \frac{\beta}{2} + r\left(\gamma - \frac{1}{2}\right) .
\end{align}
If $D = n^{o(1)}$, then we have
$$
\|L_{\leq D}\|^2 = 1+o(1).
$$
\end{proposition}

Let us consider an arbitrary subgraph $S \subseteq [M]$.
We first compute $\EE_\calP[\phi_S(Y)]$.

\begin{lemma}\label{lemma:P_expectation_low_deg}
    For $S \subseteq [M]$, let $\phi_S(Y)$ be as defined in \eqref{eqn:phi_defn}, and let $V(S) \subseteq [n]$ denote the vertex set of the hypergraph induced by $S$.
    Then we have
    \[\EE_\calP[\phi_S(Y)] = \rho^{|V(S)|}\left(\frac{p-q}{\sigma}\right)^{|S|}.\]
\end{lemma}

\begin{proof}
    We note the following:
    \[\EE_\calP[\phi_S(Y)] = \EE_Z\left[\EE_\calP[\phi_S(Y)|Z]\right].\]
    We note that $Y_e$ and $Y_f$ are conditionally independent for $e \neq f$, given $Z$. In particular, we have
    \[\EE_\calP[\phi_S(Y)|Z] = \prod_{e \in S}\EE_\calP[\tilde Y_e|Z].\]
    Furthermore, by definition of $\calP$, we have
    \[\EE_\calP[\tilde Y_e|Z] = \left\{\begin{array}{cc}
        \frac{p-q}{\sigma} & e \subseteq Z \\
        0 & e \not \subseteq Z
    \end{array}\right.\]
    With this in hand, we have
    \[\EE_\calP[\phi_S(Y)|Z] = \left(\frac{p-q}{\sigma}\right)^{|S|}\bbone{V(S) \subseteq Z},\]
    from where we get
    \[\EE_\calP[\phi_S(Y)] = \rho^{|V(S)|}\left(\frac{p-q}{\sigma}\right)^{|S|},\]
    as desired.
\end{proof}

Let $S_{\ell, m}$ be the set of all edge-induced subhypergraphs of $K_n^r$ having $\ell$ vertices and $m$ edges, or, more formally,
\begin{equation}
S_{\ell, m} \defeq \{ S \subseteq K_n^r \; : \; |V(S)| = \ell , \, |E(S)| = m\} . 
\label{eq:def-s-l-m}
\end{equation}
The following inequalities will be used throughout the proofs and hold for $n$ sufficiently large:
\begin{equation}\label{eqn:common_inequalities}
   \frac{q}{e} \leq \sigma^2 \leq q, \quad p - q \geq \frac{p}{e}.
\end{equation}

\begin{proof}[Proof of Proposition~\ref{prop:lower-large-gamma}]
Note that as the graphs in $S_{\ell, m}$ are edge induced and we only consider $S\subseteq [M]$ such that $|S| \leq D$, we have
\[r \leq \ell \leq rD, \quad \ell/r \leq m \leq D.\]
With the result of Lemma~\ref{lemma:P_expectation_low_deg}, we can conclude
\[\|L_{\leq D}\|^2 = 1 + \sum_{\ell = r}^{rD}\sum_{m = \ell/r}^D \sum_{S \in S_{\ell, m}}\rho^{2\ell} \left(\frac{p - q}{\sigma}\right)^{2m} ,\]
where the summand $1$ comes from the term $S = \varnothing$.
Let us compute an upper bound for $|S_{\ell, m}|$. 
To do so, we first consider the possible choices for the $\ell$ vertices $V(S)$ among all $n$ vertices, and then the $m$ edges $E(S)$ among all $\binom{\ell}{r}$ possible edges.
We have
\begin{align}\label{eqn:Slm_size}
    |S_{\ell, m}| \leq \binom{n}{\ell} \binom{\binom{\ell}{r}}{m} \leq n^{\ell}\binom{\ell}{r}^m \leq n^{\ell}\left(\frac{e\ell}{r}\right)^{rm} \leq n^\ell (eD)^{rm},
\end{align}
since $\ell \leq rD$.
Furthermore, as $p > q$, we have
\((p-q)^2 \leq p^2.\)
We now have
\begin{align*}
    \|L_{\leq D}\|^2 &\leq 1 + \sum_{\ell = r}^{rD}\sum_{m = \ell/r}^Dn^\ell\,(eD)^{rm}\rho^{2\ell} \left(\frac{p - q}{\sigma}\right)^{2m} \\
    &\leq 1 + \sum_{\ell = r}^{rD}(n\rho^2)^\ell\sum_{m = \ell/r}^D\left(\frac{e^{r+1}D^{r}\,p^2}{q}\right)^m \\
    &\leq 1 + \sum_{\ell = r}^{rD} n^{(2\gamma -1)\ell}\sum_{m = \ell/r}^D n^{(\beta - 2\alpha + o(1))m},
\end{align*}
where we use \eqref{eqn:common_inequalities} and the assumptions on the parameters.
By \eqref{eqn:large_gamma_lb}, we have
$\beta - 2 \alpha + o(1) < -r(2\gamma-1) - \delta$ for a fixed constant $\delta = \delta(\alpha, \beta, \gamma, r) > 0$.
It follows that
\begin{align*}
    \|L_{\leq D}\|^2 \leq 1 + \sum_{\ell = r}^{rD} n^{(2\gamma -1)\ell}\sum_{m = \ell/r}^D
    n^{(-r(2\gamma - 1) - \delta)m}.
\end{align*}
Since $\gamma \geq 1/2$, we have $n^{-r(2\gamma - 1) - \delta} < 1/2$ for $n$ sufficiently large. Therefore,
\begin{align*}
    \|L_{\leq D}\|^2 &\leq 1 + 2\sum_{\ell = r}^{rD} n^{(2\gamma -1)\ell}\,n^{(-r(2\gamma - 1) - \delta)\ell/r} = 1 + 2\sum_{\ell = r}^{rD} n^{-\delta \ell/r} \le 1 + 4 n^{-\delta},
\end{align*}
as desired.
\end{proof}

\subsection{Lower bound for small $\gamma$}\label{subsection:small_gamma_lb}

We now consider the regime $\gamma < 1/2$ and show that the desired threshold is $\alpha > \beta \gamma$.
In this case, the degree-$D$ likelihood ratio $L_{\le D}$ between $\calP$ and $\calQ$ has unbounded norm due to certain rare events under $\calP$.
To mitigate this issue, we modify our approach by conditioning on the complement of these rare events.
We then study the degree-$D$ likelihood ratio $L'_{\le D}$ between the conditional distribution $\calP'$ and $\calQ$.

\subsubsection{A high-probability event}\label{subsubsection:prob_bound}

Recall that $S_{\ell, m}$ is defined in \eqref{eq:def-s-l-m}.
Fix an arbitrarily small constant $\delta = \delta(\alpha, \beta, \gamma) > 0$ to be chosen later.
For $\ell \in \N$, define
\begin{equation}\label{eqn:m_ell}
m_\ell \defeq \left\lceil\ell\left(\frac{\gamma}{\alpha} + \delta\right) \right\rceil.
\end{equation}
Moreover, we define a set of pairs of integers
\begin{equation*}
I \defn \{ (\ell, m) \in \N^2 \; : \; m_\ell \le m \le D , \, S_{\ell,m} \ne \varnothing \} .
\end{equation*}
Let $\calC \defeq H[Z]$ denote the subhypergraph of the observed graph $H$ induced by $Z$.
In other words, $\calC$ is the planted subhypergraph to be detected.
We define the following events:
\begin{align}
    \calE_{\ell, m} &\defeq \{ \exists S\in S_{\ell, m} \text{ such that } S \subseteq \calC \} , \notag \\
    \calE &\defeq \bigcap_{(\ell, m) \in I}\calE_{\ell, m}^c , \label{eq:def-cale}
\end{align}
where $\calE_{\ell,m}^c$ is the complement of the event $\calE_{\ell,m}$.
In particular, under $\calE$, every subset of edges $S$ that induces a dense hypergraph does \emph{not} appear in $\calC$, where ``dense'' means that the edge-to-vertex ratio of $S$ is at least $\frac{\gamma}{\alpha} + \delta$.
We define a distribution $\calP'$ by 
\[\calP'(Y) \defeq \frac{\calP(Y)\,\bbone{Y \in \calE}}{\calP(\calE)}.\]
We will show $\calP(\calE) = 1 - o(1)$.
To assist with the bound, we start with the following lemma.

\begin{lemma}\label{lemma:clique_size}
    Under model $\calP$, we have
    \[\frac{n^\gamma}{2}\leq |Z| \leq \frac{3\,n^\gamma}{2},\]
    with probability at least $1 - 2\exp\left(-cn^\gamma\right)$ for some absolute constant $c > 0$.
\end{lemma}

\begin{proof}
    Note that $|Z| = \sum_{i = 1}^n z_i$ is the sum of $n$ i.i.d.\ $\BER(\rho)$ variables, i.e., $|Z| \sim \BIN(n, \rho)$. It follows from the Chernoff bound \cite[Exercise~2.3.5]{vershynin2018high} that 
    \[\calP\left(|Z| \notin \left(1 \pm \frac{1}{2}\right)\EE_\calP[|Z|]\right) \leq 2\exp\left(-C\,\frac{\EE_\calP[|Z|]}{4}\right),\]
    for some absolute constant $C > 0$.
    The lemma follows as $\EE_\calP[|Z|] = n \rho = n^\gamma$.
\end{proof}

\begin{lemma}
We have $\calP(\calE) = 1 - o(1)$.
\end{lemma}

\begin{proof}
For any events $A$ and $B$, we have
\(\P[A] \leq \P[A\mid B] + \P[B^c].\)
We let $A$ be $\calE^c$ and $B$ be the event that $|Z|$ is bounded as in Lemma \ref{lemma:clique_size}.
As a result,
\begin{equation}
\label{eq:E-complement}
\calP(\calE^c) \le \calP\left( \calE^c \; \Big| \; \frac{n^\gamma}{2}\leq |Z| \leq \frac{3\,n^\gamma}{2} \right) + 2 \exp \left( - c\,n^\gamma \right) .
\end{equation}

We now condition on a realization of $Z$ such that $\frac{n^\gamma}{2}\leq |Z| \leq \frac{3\,n^\gamma}{2}$.
We define the following set to assist with our proof:
\[S_{\ell, m}(Z) \defeq \set{S\in S_{\ell,m}\,:\, V(S) \subseteq Z}.\]
An identical argument as that in \eqref{eqn:Slm_size} shows that
\[|S_{\ell, m}(Z)| \leq \left(\frac{3}{2}n^{\gamma}\right)^\ell\,(eD)^{rm}.\]
With this in hand, we have
\[\calP(\calE_{\ell, m}\mid Z) \leq \sum_{S \in S_{\ell, m}(Z)}\calP( S \subseteq E(\calC)\mid Z) \leq \left(\frac{3}{2}n^{\gamma}\right)^\ell\,(eD)^{rm}\max_{S \in S_{\ell, m}(Z)}\calP(S \subseteq E(\calC)\mid Z).\]
For $S \in S_{\ell, m}(Z)$, as $Y_e$ and $Y_f$ are conditionally independent given $Z$, we have
\[\calP(S \subseteq E(\calC)\mid Z) = p^m = n^{-\alpha m} .\]
Plugging in this value above, we have
\[
\calP(\calE_{\ell, m}\mid Z)
\leq \left(\frac{3}{2}n^{\gamma}\right)^\ell\,(eD)^{rm}\,n^{-\alpha\,m} 
\leq \, n^{\gamma \ell - (\alpha - o(1))\,m}.
\]

Let us now bound $\calP(\calE^c\mid Z) \le \sum_{(\ell,m) \in I} \calP(\calE_{\ell,m}\mid Z)$. 
We have
\begin{align*}
    \calP(\calE^c\mid Z) \leq \sum_{\ell = r}^{rD}\sum_{m = m_\ell}^{D} n^{\gamma \ell  - (\alpha - o(1))\,m}
    = \sum_{\ell = r}^{rD} n^{\gamma \ell}
    \sum_{m = m_\ell}^{D} n^{ - (\alpha - o(1))\,m}
    \le 2 \sum_{\ell = r}^{rD} n^{\gamma \ell} n^{ - (\alpha - o(1))\,m_\ell}.
\end{align*}
By \eqref{eqn:m_ell}, we have $m_\ell \ge \ell( \frac{\gamma}{\alpha} + \delta)$, so
\begin{align*}
    \calP(\calE^c\mid Z) \leq 2 \sum_{\ell = r}^{rD} n^{-(\delta \alpha - o(1))\, \ell}
    \le 4 n^{- \delta \alpha r/2} .
\end{align*}
This combined with \eqref{eq:E-complement} completes the proof.
\end{proof}

\subsubsection{Bounding the low-degree norm}\label{subsubsection:expectation_bound}

The main result of this subsection consists in controlling the norm of the degree-$D$ likelihood ratio between $\calP'$ and $\calQ$ defined in \eqref{eq:conditional-degree-d-norm}.

\begin{proposition}
\label{prop:lower-small-gamma}
Suppose $p = n^{-\alpha}$, $q = n^{-\beta}$, and $\rho = n^{\gamma - 1}$ with fixed $r \ge 2$, $0 < \alpha < \beta < r-1$, and $\gamma \in (0,1)$ such that
\begin{align}\label{eqn:small_gamma_lb}
    \gamma < \frac{1}{2}, \quad \alpha > \beta\,\gamma .
\end{align}
If $D = n^{o(1)}$, then we have
$$
\|L_{\leq D}'\|^2 = 1+o(1).
$$
\end{proposition}

To prove the proposition, we start from \eqref{eq:conditional-degree-d-norm}.
Note the following as a result of \S\ref{subsubsection:prob_bound}:
\begin{equation}
\EE_{\calP'}[\phi_S(Y)] = \frac{1}{\calP(\calE)}\,\EE_{\calP}[\phi_S(Y)\bbone{Y \in \calE}] = (1+o(1))\EE_{\calP}[\phi_S(Y)\bbone{Y \in \calE}]. \label{eq:p-prime-to-p-exp}
\end{equation}
It is now enough to consider the final term.
We first condition on the outcome $Z$ to get
\begin{align*}
    \EE_\calP[\phi_S(Y)\bbone{Y \in \calE}] &= \EE_Z[\EE_\calP[\phi_S(Y)\bbone{Y \in \calE} \mid Z]].
\end{align*}
Consider a realization $Z$.
Suppose there exists $e \in S$ such that $e \not \subseteq Z$. 
Since $\calE$ is determined by the edges in $\calC$, $Y_e$ is conditionally independent of $\phi_{S\setminus\{e\}}(Y)\bbone{Y \in \calE}$, given $Z$.
In particular, we have
\[\EE_\calP[\phi_S(Y)\bbone{Y \in \calE} \mid Z] = \EE_\calP[\tilde Y_e \mid Z] \,\EE_\calP[\phi_{S\setminus\{e\}}(Y)\bbone{Y \in \calE} \mid Z] = 0,\]
as $\EE_\calP[\tilde Y_e \mid Z] = 0$ for $e \not\subseteq Z$.
It follows that
\begin{align}\label{eqn:expectation_bound}
    \EE_\calP[\phi_S(Y)\bbone{Y \in \calE}] &= \EE_Z[\EE_\calP[\phi_S(Y)\bbone{Y \in \calE} \mid Z]\,\bbone{V(S) \subseteq Z}].
\end{align}
We split the set of $S \subseteq K_n^r$ with $|S| \le D$ into a set of ``bad'' subgraphs and a set of ``good'' subgraphs, defined respectively as
$$
\mathcal B \defeq \bigcup_{(\ell, m) \in I} S_{\ell,m} , \quad
\mathcal G \defeq \{S \subseteq K_n^r \; : \; |S| \le D\} \setminus \mathcal B .
$$
Let us first consider a good subgraph $S$. We prove the following lemma.

\begin{lemma}\label{lemma:good_expectation_bound}
    Let $S \in \mathcal G$ be a good subgraph on $\ell$ vertices and $m$ edges. We have
     \[|\EE_\calP[\phi_S(Y)\bbone{Y \in \calE}]| \leq \rho^\ell\left(\frac{2p}{\sigma}\right)^m.\]
\end{lemma}

\begin{proof}
    Note the following as a result of \eqref{eqn:expectation_bound} and Jensen's inequality:
    \begin{equation}|\EE_\calP[\phi_S(Y)\bbone{Y \in \calE}]| \leq \EE_Z[\EE_\calP[|\phi_S(Y)|\bbone{Y \in \calE} \mid Z]\,\bbone{V(S) \subseteq Z}]. \label{eq:phi-s-jensen}
    \end{equation}
    Let us consider $\EE_\calP[|\phi_S(Y)|\bbone{Y \in \calE} \mid Z]$ for $V(S) \subseteq Z$.
    Since $|\phi_S(Y)| \geq 0$ we have
    \begin{align*}
        \EE_\calP[|\phi_S(Y)|\bbone{Y \in \calE} \mid Z] &\leq \EE_\calP[|\phi_S(Y)|\mid Z] 
        = \left(\frac{p(1-q) + (1-p)q}{\sigma}\right)^m 
        \leq \left(\frac{2p}{\sigma}\right)^m,
    \end{align*}
    where the last step follows since $q < p$.
    We now have:
    \begin{align*}
        |\EE_\calP[\phi_S(Y)\bbone{Y \in \calE}]|
        \leq \left(\frac{2p}{\sigma}\right)^m\EE_Z\left[\bbone{V(S) \subseteq Z}\right] 
        = \rho^\ell\left(\frac{2p}{\sigma}\right)^m,
    \end{align*}
    as desired.
\end{proof}

Next, let us consider a bad subgraph $S$. We will prove the following lemma.

\begin{lemma}\label{lemma:bad_expectation_bound}
    Let $S \in \mathcal B$ be a bad subgraph on $\ell$ vertices and $m$ edges. We have
    \[|\EE_\calP[\phi_S(Y)\bbone{Y \in \calE}]| \leq \rho^\ell\,\binom{m}{m_\ell-1}\,\frac{q^{m - m_\ell + 1}\,(2p)^{m_\ell -1}}{\sigma^m}.\]
\end{lemma}

\begin{proof}
    As in the proof of Lemma \ref{lemma:good_expectation_bound}, we will first bound $\EE_\calP[|\phi_S(Y)|\bbone{Y \in \calE} \mid Z]$ for $V(S) \subseteq Z$.
    By the definition of $\calE$ in \eqref{eq:def-cale}, $Y \in \calE$ implies that $\calC$ contains at most $m_\ell - 1$ edges in $S \in \mathcal B$ which has $m$ edges.
    Therefore, for $s \defeq m - m_\ell + 1 \ge 1$, there exist $e_1, \ldots, e_s \in S$ such that $Y_{e_i} = 0$ for each $i \in [s]$.
    With this in mind, we have
    \begin{align*}
        \EE_\calP[|\phi_S(Y)|\bbone{Y \in \calE} \mid Z] 
        &\leq \sum_{e_1, \ldots, e_s \in S}\EE_\calP[|\phi_S(Y)|\bbone{Y_{e_1} = \ldots = Y_{e_s} = 0}\mid Z] \\
        &\leq \sum_{e_1, \ldots, e_s \in S}\left(\frac{q}{\sigma}\right)^s\left(\frac{p(1-q) + (1-p)q}{\sigma}\right)^{m - s} 
        \leq \binom{m}{s}\,\frac{q^s\,(2p)^{m-s}}{\sigma^m},
    \end{align*}
    where the last step follows since $q < p$.
    By \eqref{eq:phi-s-jensen} again, we obtain
    \begin{align*}
        |\EE_\calP[\phi_S(Y)\bbone{Y \in \calE}]| 
        &\le \binom{m}{s}\,\frac{q^s\,(2p)^{m-s}}{\sigma^m}\EE_Z\left[\bbone{V(S) \subseteq Z}\right] 
        = \rho^\ell\binom{m}{s}\,\frac{q^s\,(2p)^{m-s}}{\sigma^m},
    \end{align*}
    as desired. 
\end{proof}
    
Recall our original goal of bounding \eqref{eq:conditional-degree-d-norm}.
By \eqref{eq:p-prime-to-p-exp},
we can simplify \eqref{eq:conditional-degree-d-norm} to
\begin{align}
    \|L_{\leq D}'\|^2 
    &= (1+o(1)) \sum_{\substack{S \subseteq [M], \\ |S| \leq D}}\EE_{\calP}[\phi_S(Y)\bbone{Y \in \calE}]^2 \nonumber \\
    &= (1+o(1))\left(\sum_{S \in \mathcal G}\EE_{\calP}[\phi_S(Y)\bbone{Y \in \calE}]^2 + \sum_{S \in \mathcal B}\EE_{\calP}[\phi_S(Y)\bbone{Y \in \calE}]^2\right). \label{eqn:good_bad_partition}
\end{align}
Let us consider the first sum in \eqref{eqn:good_bad_partition}. 

\begin{lemma}\label{lemma:good_bound}
    We have $\sum_{S \in \mathcal G}\EE_{\calP}[\phi_S(Y)\bbone{Y \in \calE}]^2 = 1 + o(1)$.
\end{lemma}

\begin{proof}
    If $S$ is the empty graph, then $\E_\calP[\phi_S(Y)] = 1$.
    For nonempty $S \in S_{\ell,m} \subseteq \mathcal G$, we have
    \[\frac{\ell}{r} \leq m \leq m_{\ell} - 1,\]
    since $S$ has no isolated vertices.
    From \eqref{eqn:Slm_size} and Lemma \ref{lemma:good_expectation_bound}, we obtain
    \begin{align*}
        \sum_{S \in \mathcal G}\EE_{\calP}[\phi_S(Y)\bbone{Y \in \calE}]^2 &\leq 1 + \sum_{\ell = r}^{rD}\sum_{m = \ell/r}^{\min(m_\ell - 1, D)}\sum_{S\in S_{\ell, m}}\rho^{2\ell}\left(\frac{2p}{\sigma}\right)^{2m} \\
        &\leq 1 + \sum_{\ell = r}^{rD}\sum_{m = \ell/r}^{\min(m_\ell - 1, D)}n^\ell(eD)^{rm}\rho^{2\ell}\left(\frac{2p}{\sigma}\right)^{2m} \\
        &\leq 1 + \sum_{\ell = r}^{rD}\sum_{m = \ell/r}^{\min(m_\ell - 1, D)}(n\rho^2)^\ell\left(\frac{e^{r+3}D^r\,p^2}{q}\right)^m \\
        &= 1 + \sum_{\ell = r}^{rD}n^{(2\gamma - 1)\ell}\sum_{m = \ell/r}^{\min(m_\ell - 1, D)}n^{(\beta - 2\alpha + o(1))m},
    \end{align*}
    where we use \eqref{eqn:common_inequalities}.
    We now consider two cases. First, suppose $\beta < 2\alpha$.
    Then, we have for some small enough constant $\delta_1 = \delta_1(\alpha, \beta) > 0$ and $n$ large enough:
    \begin{align*}
        \sum_{\ell = r}^{rD}n^{(2\gamma - 1)\ell}\sum_{m = \ell/r}^{\min(m_\ell - 1 , D)}n^{(\beta - 2\alpha + o(1))m} \leq \sum_{\ell = r}^{rD}n^{(2\gamma - 1)\ell}\sum_{m = \ell/r}^{\min(m_\ell - 1 , D)}n^{-\delta_1 m} 
        \leq 2\sum_{\ell = r}^{rD}n^{(2\gamma - 1)\ell - \delta_1 \ell/r}.
    \end{align*}
    Since $\gamma < 1/2$, we can simplify the above further to get
    \[2\sum_{\ell = r}^{rD}n^{(2\gamma - 1)\ell - \delta_1 \ell/r} < 2\sum_{\ell = r}^{rD}n^{- \delta_1\ell/r} \leq 4n^{-\delta_1},\]
    as desired.    
    Now, let us assume $\beta \geq 2\alpha$.
    We have
    \begin{align}
        \sum_{\ell = r}^{rD}n^{(2\gamma - 1)\ell}\sum_{m = \ell/r}^{\min(m_\ell - 1 , D)}n^{(\beta - 2\alpha + o(1))m} 
        &\leq \sum_{\ell = r}^{rD} n^{(2\gamma - 1)\ell}\left(D - \frac{\ell}{r} + 1\right)n^{(\beta - 2\alpha + o(1))(m_\ell-1)} \notag \\
        &\le \sum_{\ell = r}^{rD} n^{(2\gamma - 1)\ell + (\beta - 2\alpha + o(1))(m_\ell-1) + o(1)}. \label{eq:exponent-to-be-bounded}
    \end{align}
    Let us now consider the coefficient of $m_\ell-1$ in the exponent of $n$. 
    By the definition of $m_\ell$ in \eqref{eqn:m_ell}, we have
    \[m_\ell - 1 < \ell\left(\frac{\gamma}{\alpha} + \delta \right).\]
    With this in hand, we have 
    \begin{align*}
        (\beta - 2\alpha + o(1))(m_\ell-1)
        < (\beta - 2\alpha + o(1))\left(\frac{\gamma}{\alpha} + \delta\right)\ell  
        < \left(\frac{\beta\gamma}{\alpha} - 2\gamma + r\delta\right)\ell .
    \end{align*}
    From here it follows that the exponent in \eqref{eq:exponent-to-be-bounded} can be bounded as
    \begin{equation} 
    (2\gamma - 1)\ell + (\beta - 2\alpha + o(1))(m_\ell-1) + o(1) \le \left(\frac{\beta\gamma}{\alpha} - 1 + 2r\delta\right)\ell 
    \le - \delta \ell , 
    \label{eq:exponent-bound-1}
    \end{equation}
   where the last step holds in view of the condition $\alpha > \beta \gamma$ in \eqref{eqn:small_gamma_lb}, once we choose $\delta = \delta(\alpha, \beta, \gamma) > 0$ to be sufficiently small.
    Then \eqref{eq:exponent-to-be-bounded} can be further bounded as
    \[\sum_{\ell = r}^{rD} n^{(2\gamma - 1)\ell + (\beta - 2\alpha + o(1))(m_\ell-1) + o(1)} \leq \sum_{\ell = r}^{rD} n^{-\delta\ell} \leq 2n^{-\delta\,r},\]
    as desired.
\end{proof}

Now, let us consider the bad subgraphs and bound the second sum in \eqref{eqn:good_bad_partition}.

\begin{lemma}\label{lemma:bad_bound}
We have $\sum_{S \in \mathcal B}\EE_{\calP}[\phi_S(Y)\bbone{Y \in \calE}]^2 = o(1)$.
\end{lemma}

\begin{proof}
    Note that for $S \in S_{\ell,m} \subseteq \mathcal B$, we have
    \[m_{\ell} \leq m \leq D.\]
    From \eqref{eqn:Slm_size} and Lemma \ref{lemma:bad_expectation_bound}, we have
    \begin{align*}
        \sum_{S \in \mathcal B}\EE_{\calP}[\phi_S(Y)\bbone{Y \in \calE}]^2 &\leq \sum_{\ell = r}^{rD}\sum_{m = m_\ell}^{D}\sum_{S\in S_{\ell, m}}\rho^{2\ell}\binom{m}{m_\ell-1}^2\,\frac{q^{2(m - m_\ell + 1)}\,(2p)^{2(m_\ell -1)}}{\sigma^{2m}} \\    
        &\leq  \sum_{\ell = r}^{rD}\sum_{m = m_\ell}^{D}n^\ell(eD)^{rm}\rho^{2\ell}\left(\frac{q}{\sigma}\right)^{2(m - m_\ell + 1)}\left(\frac{2\,m\,p}{\sigma}\right)^{2(m_\ell - 1)} \\
        &\leq \sum_{\ell = r}^{rD}\sum_{m = m_\ell}^{D}(n\rho^2)^\ell\left(\frac{e^r\,D^r\,q^2}{\sigma^2}\right)^m\left(\frac{2\,m\,p}{q}\right)^{2(m_\ell - 1)} \\
        &= \sum_{\ell = r}^{rD}n^{(2\gamma - 1)\ell}\,n^{2(\beta - \alpha + o(1))(m_{\ell} - 1)}\sum_{m =m_\ell}^{D}n^{(-\beta + o(1))m},
    \end{align*}
    where we use \eqref{eqn:common_inequalities}. 
    Note that the coefficient of $m$ in the exponent is negative for large enough $n$. 
    Thus, we can further simplify the above to
    \begin{align*}
        \sum_{S \in \mathcal B}\EE_{\calP}[\phi_S(Y)\bbone{Y \in \calE}]^2 &\leq 2\sum_{\ell = r}^{rD}n^{(2\gamma - 1)\ell}\,n^{2(\beta - \alpha + o(1))(m_{\ell} - 1)}\,n^{(-\beta + o(1))m_\ell} \\
        &= 2\sum_{\ell = r}^{rD}n^{(2\gamma - 1)\ell + (\beta - 2\alpha + o(1))(m_{\ell} - 1) - \beta + o(1)}.
    \end{align*}
    Once again, we consider two cases.
    First, let us assume $\beta < 2\alpha$.
    As $\gamma < 1/2$, we have for $n$ large enough:
    \[\sum_{S \in \mathcal B}\EE_{\calP}[\phi_S(Y)\bbone{Y \in \calE}]^2 \leq 2\sum_{\ell = r}^{rD}n^{ - \beta + o(1)} \leq 2rDn^{- \beta + o(1)} \leq n^{-\beta/2},\]
    as desired.
    Now, let us assume $\beta \geq 2\alpha$.
    An identical argument as in the proof of \eqref{eq:exponent-bound-1} shows that
    \[(2\gamma - 1)\ell + (\beta - 2\alpha + o(1))(m_{\ell} - 1) - \beta + o(1) \le -\delta\ell - \beta .\]
    In particular, we have
    \[\sum_{S \in \mathcal B}\EE_{\calP}[\phi_S(Y)\bbone{Y \in \calE}]^2 \leq 2\sum_{\ell = r}^{rD}n^{-\delta\ell- \beta} \leq 4n^{-\beta - \delta r},\]
    as desired.
\end{proof}

\begin{proof}[Proof of Proposition~\ref{prop:lower-small-gamma}]
It suffices to combine \eqref{eqn:good_bad_partition} with Lemmas \ref{lemma:good_bound} and \ref{lemma:bad_bound}.
\end{proof}

\section{Detection upper bound}\label{section:upper_bound}

As in the previous section, we split the analysis into two cases based on the value $\gamma$.
For $\gamma \ge 1/2$, we consider the signed count of the total number of hyperedges in $H$, as defined in \eqref{eq:def-stat-1}.
In particular, it is a linear function in the entries of the adjacency tensor $Y$.
For $\gamma < 1/2$, our statistic is the (unsigned) count of the occurrences of a certain \emph{balanced} subhypergraph in $H$, as defined in \eqref{eq:def-t-statistic-2}.
The size of the subhypergraph depends only on the fixed parameters $(\alpha, \beta, \gamma)$ as given in Proposition~\ref{prop:existence-balanced-hypergraph}.
Therefore, the statistic in \eqref{eq:def-t-statistic-2} is a constant-degree polynomial in the entries of $Y$ and can be computed in polynomial time.

\begin{proof}[Proof of Theorem~\ref{thm:main-results} (upper bound)]
If $\gamma \ge 1/2$ and $\alpha < \beta/2 + r(\gamma - 1/2)$, Proposition~\ref{prop:upper-large-gamma} shows that the statistic $\tilde{T}$ defined in \eqref{eq:def-stat-1} strongly separates $\calP$ and $\calQ$.
If $\gamma < 1/2$ and $\alpha < \beta \gamma$, Proposition~\ref{prop:upper-small-gamma} shows that the statistic $T$ defined in \eqref{eq:def-t-statistic-2} strongly separates $\calP$ and $\calQ$.
\end{proof}

\subsection{Upper bound for large $\gamma$}\label{subsection:large_gamma_ub}

In the regime where $\gamma \ge 1/2$, we will use the signed count of hyperedges in $H$ as our test statistic.
Recall $M = \binom{n}{r}$ as defined in \S \ref{section:main_results}, and $[M]$ denotes the set of all possible hyperedges.
We define
\begin{equation}
\tilde{T} \defeq \sum_{e\in [M]}\tilde Y_e,
\label{eq:def-stat-1}
\end{equation}
where $\tilde Y_e$ is defined in \eqref{eq:def-tilde-y-e}. Trivially, $\tilde{T}$ is a degree-$1$ polynomial of the entries of $Y$.

\begin{proposition}
\label{prop:upper-large-gamma}
Suppose $p = n^{-\alpha}$, $q = n^{-\beta}$, and $\rho = n^{\gamma - 1}$ with fixed $r \ge 2$, $0 < \alpha < \beta < r-1$, and $\gamma \in (0,1)$ such that
\begin{align}\label{eqn:large_gamma_ub}
    \gamma \geq \frac{1}{2}, \quad \alpha < \frac{\beta}{2} + r\left(\gamma - \frac{1}{2}\right).
\end{align}
For the statistic $\tilde{T}$ defined in \eqref{eq:def-stat-1}, we have
\( \sqrt{\Var_{\calP}(\tilde{T}) \vee \Var_{\calQ}(\tilde{T})} = o\left(\left|\E_{\calP}[\tilde{T}] - \E_{\calQ}[\tilde{T}]\right|\right) .\)
\end{proposition}

Let us first compute the expectation and variance under each distribution.

\begin{lemma}\label{lemma:num_edges_expectations}
    We have
    \begin{align*}
        \EE_{\calQ}[\tilde{T}] &= 0, \quad \Var_\calQ(\tilde{T}) = M, \quad \EE_\calP[\tilde{T}] = M\,\rho^r\left(\frac{p-q}{\sigma}\right), \\
        \Var_\calP(\tilde{T}) &\leq M + \frac{2M\rho^r\,p}{\sigma^2} + \frac{2\,M\,rn^{r-1}\rho^{2r-1}\,p^2}{\sigma^2}.
    \end{align*}
\end{lemma}

\begin{proof}
    Note that $\tilde Y_e$ is defined by standardizing $Y_e$ under the distribution $\calQ$, and $Y_e$ and $Y_f$ are independent under $\calQ$ for $e \neq f$. 
    The results on $\EE_\calQ[\tilde{T}]$ and $\Var_\calQ(t)$ follow immediately. 
    
    Let us now consider an edge $e \in [M]$.
    We have
    \[\EE_\calP[\tilde Y_e] = \rho^r\left(\frac{p-q}{\sigma}\right) + (1 - \rho^r)\cdot 0 = \rho^r\left(\frac{p-q}{\sigma}\right),\]
    which leads to the desired formula for $\EE_\calP[\tilde{T}]$.
    
    Now, we have the following for the variance under $\calP$:
    \[\Var_\calP(\tilde{T}) = \sum_{e, f \in [M]}\Cov_\calP(\tilde Y_e, \tilde Y_f).\]
    Notice that $\tilde Y_e$ and $\tilde Y_f$ are independent for $e\cap f = \varnothing$.
    It follows that
    \begin{align}
        \Var_\calP(\tilde{T}) &= \sum_{\substack{e, f \in [M], \\ e\cap f \neq \varnothing}}\Cov_\calP(\tilde Y_e, \tilde Y_f) \notag \\        
        &= \sum_{\substack{e, f \in [M], \\ e\cap f \neq \varnothing}}\left(\EE_\calP[\tilde Y_e \tilde Y_f] - \rho^{2r}\left(\frac{p-q}{\sigma}\right)^2\right) \notag \\ 
        &\leq \sum_{e \in [M]}\EE_\calP[\tilde Y_e^2] + \sum_{e \in [M]}\sum_{\substack{f \in [M] \setminus e, \\ e\cap f \neq \varnothing}}\EE_\calP[\tilde Y_e \tilde Y_f] . \label{eq:var-t-p-two-terms}
    \end{align}
    Let us consider the first term.
    We have
    \begin{align*}
        \EE_\calP[\tilde Y_e^2] 
        &= \rho^r\left(p\left(\frac{1 - q}{\sigma}\right)^2 + (1-p)\left(\frac{q}{\sigma}\right)^2\right) + (1-\rho^r)\left(q\left(\frac{1 - q}{\sigma}\right)^2 + (1-q)\left(\frac{q}{\sigma}\right)^2\right) \\
        &\leq \rho^r\left(\frac{p(1-q)^2 + (1-p)q^2}{\sigma^2}\right) + 1 \\
        &\leq \frac{2\rho^r\,p}{\sigma^2} + 1,
    \end{align*}
    as $q^2 < q < p$.
    For the second term in \eqref{eq:var-t-p-two-terms}, let us first compute $\EE_{\calP}[\tilde Y_e\tilde Y_f]$.
    As $\tilde Y_e$ and $\tilde Y_f$ are independent given $Z$, we have
    \[\EE_{\calP}[\tilde Y_e\tilde Y_f] = \EE_Z[\EE_{\calP}[\tilde Y_e\tilde Y_f\mid Z]] = \EE_Z[\EE_{\calP}[\tilde Y_e\mid Z]\EE_{\calP}[\tilde Y_f\mid Z]].\]
    We note the following:
    \[\EE_{\calP}[\tilde Y_e\mid Z] = \left(\frac{p-q}{\sigma}\right)\bbone{e \subseteq Z} + 0\cdot\bbone{e \not\subseteq Z} = \left(\frac{p-q}{\sigma}\right)\bbone{e \subseteq Z}.\]
    In particular, we have
    \[\EE_{\calP}[\tilde Y_e\tilde Y_f] = \EE_Z\left[\left(\frac{p-q}{\sigma}\right)^2\bbone{e \cup f \subseteq Z}\right] = \rho^{2r - |e\cap f|}\left(\frac{p-q}{\sigma}\right)^2 \leq \frac{\rho^{2r - |e\cap f|}\,p^2}{\sigma^2},\]
    where the last step follows since $p > q$.
    Putting these bounds together with \eqref{eq:var-t-p-two-terms}, we now have
    \[\Var_\calP(\tilde{T}) \leq M + \frac{2M\rho^r\,p}{\sigma^2} + \frac{\rho^{2r}\,p^2}{\sigma^2}\sum_{e \in [M]}\sum_{\substack{f \in [M] \setminus e, \\ e\cap f \neq \varnothing}}\rho^{-|e\cap f|}.\]
    Let us fix $e$.
    We will consider bounding 
    \[\sum_{\substack{f \in [M] \setminus e, \\ e\cap f \neq \varnothing}}\rho^{-|e\cap f|} \leq \sum_{s = 1}^{r - 1}\sum_{v_1, \ldots, v_s \in e}\sum_{\substack{f \in [M], \\ v_1, \ldots, v_s \in f}}\rho^{-s}.\]
    For a fixed $v_1, \ldots, v_s$, there are at most $\binom{n}{r-s} \leq n^{r-s}$ choices for $f$.
    Similarly, there are at most $\binom{r}{s} \leq r^s$ choices for $v_1, \ldots, v_s$.
    It follows that
    \[\sum_{\substack{f \in [M] \setminus e, \\ e\cap f \neq \varnothing}}\rho^{-|e\cap f|} \leq \sum_{s = 1}^{r - 1}r^s\,n^{r-s}\,\rho^{-s} = n^r\sum_{s = 1}^{r - 1}\left(\frac{r}{n\rho}\right)^s.\]
    Note that $n\rho = n^\gamma \gg r$ and so we have
    \[\sum_{\substack{f \in [M] \setminus e, \\ e\cap f \neq \varnothing}}\rho^{-|e\cap f|} \leq 2rn^{r-1}\rho^{-1}.\]
    Plugging this value in for our variance computation, we get
    \[\Var_\calP(\tilde{T}) \leq M + \frac{2M\rho^r\,p}{\sigma^2} + \frac{2\,M\,rn^{r-1}\rho^{2r-1}\,p^2}{\sigma^2},\]
    as desired.   
\end{proof} 
    
\begin{proof}[Proof of Proposition~\ref{prop:upper-large-gamma}]

Lemma~\ref{lemma:num_edges_expectations} implies that it is enough to show
\begin{align*}
    M^2\rho^{2r}\left(\frac{p-q}{\sigma}\right)^2 &\gg M + \frac{2M\rho^r\,p}{\sigma^2} + \frac{2\,M\,rn^{r-1}\rho^{2r-1}\,p^2}{\sigma^2} \\
    \iff \rho^{2r}\left(\frac{p-q}{\sigma}\right)^2 &\gg \frac{1}{M} + \frac{2\rho^r\,p}{M\sigma^2} + \frac{2\,rn^{r-1}\rho^{2r-1}\,p^2}{M\sigma^2}.
\end{align*}
Let us first lower bound the left-hand side. 
From \eqref{eqn:common_inequalities}, we have
\[LHS \geq \frac{\rho^{2r}\,p^2}{e^2q} = e^{-2} \, n^{2r(\gamma - 1) + \beta - 2\alpha} .\]
Next, let us upper bound the right-hand side.
Note that 
\(M = \binom{n}{r} \geq \left(\frac{n}{r}\right)^r\)
so that $\frac{1}{M} = n^{- r + o(1)}$.
From \eqref{eqn:common_inequalities}, we have
\[RHS \leq n^{- r + o(1)} + n^{r(\gamma - 2) + \beta - \alpha + o(1)} + n^{-1 + (2r-1)(\gamma - 1) + \beta - 2\alpha + o(1)} .\]
It remains to verify that
\begin{align*}
2r(\gamma - 1) + \beta - 2\alpha &> - r + o(1) , \\
2r(\gamma - 1) + \beta - 2\alpha &> r(\gamma - 2) + \beta - \alpha + o(1) , \\
2r(\gamma - 1) + \beta - 2\alpha &> -1 + (2r-1)(\gamma - 1) + \beta - 2\alpha + o(1) ,
\end{align*}
using the assumption \eqref{eqn:large_gamma_ub}, which is straightforward.
\end{proof}

\subsection{Upper bound for small $\gamma$}\label{subsection:small_gamma_ub}

In this subsection, we will assume 
\begin{align}\label{eqn:small_gamma_ub}
    \gamma < \frac 12, \quad \alpha < \beta\,\gamma.
\end{align}
As in the previous subsection, we will define a test statistic $T$ to distinguish between $\calP$ and $\calQ$.

\subsubsection{Balanced hypergraphs}

The statistic $T$ we choose will count the occurrences of a specific hypergraph $\tilde H$ in $H$.
We first define balanced hypergraphs.

\begin{definition}\label{def:balanced_hypergraphs}
    A hypergraph $H$ is balanced if for every nonempty $H' \subseteq H$, we have
    \[\frac{|E(H')|}{|V(H')|} \leq \frac{|E(H)|}{|V(H)|}.\]
\end{definition}

Let us now define our hypergraph of interest $\tilde H$, whose existence is guaranteed by \cite{rucinski1986strongly} for $r = 2$ and by \cite{matushkin2022strictly} in the general case.

\begin{proposition}
\label{prop:existence-balanced-hypergraph}
Assuming \eqref{eqn:small_gamma_ub}, there exists a balanced hypergraph $\tilde H$ with $\ell$ vertices and $m$ edges satisfying
\begin{equation}
\frac{1}{\beta} < \frac{m}{\ell} < \frac{\gamma}{\alpha}.  \label{eq:edge-vertex-ratio-between} 
\end{equation}
\end{proposition}

\begin{proof}
Since $\alpha < \beta \gamma$ by assumption, there is a rational number $\lambda$ such that $\frac{1}{\beta} < \lambda < \frac{\gamma}{\alpha}$.
Furthermore, note that
\(\lambda > \frac{1}{\beta} \geq \frac{1}{r - 1}.\)
By \cite[Theorem~2]{matushkin2022strictly}, there exists a balanced hypergraph with $\ell$ vertices and $m$ edges such that $\frac{m}{\ell} = \lambda$, completing the proof.
\end{proof}

Similar to \eqref{eqn:phi_defn}, we define 
\begin{equation*}
    \psi_S(Y) \defeq \prod_{e \in S}Y_e, \quad \forall S \subseteq [M]. 
\end{equation*}
With $\tilde H$ in hand, we are ready to define $T$:
\begin{equation}
    T \defeq \sum_{S\in \calS} \psi_{S}(Y) , \, \quad \text{ where } \calS \defeq \set{S \subseteq [M]\,:\, K_n^r[S] \cong \tilde H} .
    \label{eq:def-t-statistic-2}
\end{equation}
In particular, $T$ counts the number of edge-induced subgraphs in $H$ isomorphic to $\tilde H$.
Since the number of edges $m$ in $\tilde H$ is a constant depending only on $(\alpha, \beta, \gamma)$, the statistic $T$ is a constant-degree polynomial in the entries of $Y$. 
We remark that $T$ is taken to be the (unsigned) count rather than the signed count mainly for technical convenience.
The fact that $\psi_S(Y)$ takes value only in $\{0,1\}$ simplifies the computation of the variance of $T$.

Before turning to the main result of this section, we note a simple fact about balanced hypergraphs that will be useful in the proof later.

\begin{claim}\label{claim:balanced_hypergraph_complement}
    Let $H$ be a balanced hypergraph and let $H'$ be a proper subhypergraph of $H$ with $V(H') \subsetneq V(H)$.
    Then,
    \[\frac{|E(H)| - |E(H')|}{|V(H)| - |V(H')|} \geq \frac{|E(H)|}{|V(H)|}.\]
\end{claim}

\begin{proof}
    If $H' = \varnothing$, the inequality holds trivially. 
    Therefore, we may assume $H'$ is nonempty.
    As $H$ is balanced and $H'$ is nonempty, by Definition \ref{def:balanced_hypergraphs} we have
    \begin{align*}
        \frac{|E(H)| - |E(H')|}{|V(H)| - |V(H')|} \geq \frac{|E(H)| - |V(H')|\frac{|E(H)|}{|V(H)|}}{|V(H)| - |V(H')|} 
        = \frac{|E(H)|(|V(H)| - |V(H')|)}{|V(H)|(|V(H)| - |V(H')|)} 
        = \frac{|E(H)|}{|V(H)|},
    \end{align*}
    as desired.
\end{proof}

\subsubsection{Strong separation}

We now show that the statistic $T$ strongly separates $\calP$ and $\calQ$ in the regime of interest.

\begin{proposition}
\label{prop:upper-small-gamma}
Suppose $p = n^{-\alpha}$, $q = n^{-\beta}$, and $\rho = n^{\gamma - 1}$ with fixed $r \ge 2$, $0 < \alpha < \beta < r-1$, and $\gamma \in (0,1)$ such that \eqref{eqn:small_gamma_ub} holds.
The statistic $T$ defined in \eqref{eq:def-t-statistic-2} satisfies
\( \sqrt{\Var_{\calP}(T) \vee \Var_{\calQ}(T)} = o\left(\left|\E_{\calP}[T] - \E_{\calQ}[T]\right|\right) .\)
\end{proposition}

Let us first consider $T$ under $\calQ$.
We have the following lemma.

\begin{lemma}\label{lemma:Q_bounds_T}
    With $T$ as defined in \eqref{eq:def-t-statistic-2}, we have
    \[\EE_\calQ[T] = N\,q^m, \quad \Var_\calQ(T) \leq \max\left\{m^2\,N\,n^{\ell(1 - 1/m)}\,\ell^{r(m-1)}\,q^{2m - 1}, m^{m+1}\EE_\calQ[T]\right\},\]
    where $N = |\calS|$.
\end{lemma}

\begin{proof}
    As $Y_e$ and $Y_f$ are independent for $e \neq f$ under $\calQ$, we note that for any $S \subseteq [M]$, we have
    \[\psi_{S}(Y) \sim \BER(q^m),\]
    from where the result on $\EE_\calQ[T]$ follows.
    For the variance, we have
    \[\Var_\calQ(T) = \sum_{S_1, S_2 \in \calS}\Cov_\calQ\left(\psi_{S_1}(Y), \psi_{S_2}(Y)\right).\]
    Note that $\psi_{S_1}(Y)$ and $\psi_{S_2}(Y)$ are independent if $S_1\cap S_2 = \varnothing$.
    Furthermore, we have 
    \begin{align*}
        \psi_{S_1}(Y)\psi_{S_2}(Y) = \prod_{e \in S_1}Y_e\prod_{e \in S_2}Y_e = \prod_{e \in S_1\cap S_2}Y_e^2 \prod_{e \in S_1\Delta S_2}Y_e = \prod_{e \in S_1\cap S_2}Y_e\prod_{e \in S_1\Delta S_2}Y_e = \prod_{e \in S_1\cup S_2}Y_e,
    \end{align*}
    as $Y_e^2 = Y_e$. In particular, we can bound the variance by
    \begin{align}
        \Var_\calQ(T) &= \sum_{\substack{S_1, S_2 \in \calS, \\ S_1 \cap S_2 \neq \varnothing}}\Cov_\calQ\left(\psi_{S_1}(Y), \psi_{S_2}(Y)\right) \notag \\
        &= \sum_{\substack{S_1, S_2 \in \calS, \\ S_1 \cap S_2 \neq \varnothing}} \left( \EE_\calQ\left[\psi_{S_1}(Y)\psi_{S_2}(Y)\right] - \EE_\calQ\left[\psi_{S_1}(Y)\right] \, \EE_\calQ\left[\psi_{S_2}(Y)\right] \right) \notag \\
        &\leq \sum_{\substack{S_1, S_2 \in \calS, \\ S_1 \cap S_2 \neq \varnothing}}\EE_\calQ\left[\psi_{S_1\cup S_2}(Y)\right] \label{eq:var-s-1-s-2-union} \\
        &= \sum_{\substack{S_1, S_2 \in \calS, \\ S_1 \cap S_2 \neq \varnothing}}q^{|S_1\cup S_2|} = q^{2m}\sum_{\substack{S_1, S_2 \in \calS, \\ S_1 \cap S_2 \neq \varnothing}}q^{-|S_1\cap S_2|} . \notag
    \end{align}
    Let us fix $S_1$.
    We will consider bounding
    \begin{align}\label{eqn:Q_bound_for_S1}
        \sum_{\substack{S_2 \in \calS, \\ S_1 \cap S_2 \neq \varnothing}}q^{-|S_1\cap S_2|} \leq \sum_{m' = 1}^m\sum_{\substack{S' \subseteq S_1, \\ |S'| = m'}}\sum_{\substack{S_2 \in \calS, \\ S' \subseteq S_2}}q^{-m'} .
    \end{align}
    Let us first consider a fixed $S'$ and bound the number of choices for $S_2$ in the inner-most sum.
    There are at most $n^{\ell - |V(S')|}$ choices for the remaining vertices of $S_2$ and at most $\ell^{r(m-m')}$ choices for the remaining edges $S_2 \setminus S'$.
    We note that, as $S' \subseteq S_1$ and $K_n^r[S_1]$ is isomorphic to a balanced hypergraph $\tilde H$, we have
    \[\frac{|S'|}{|V(S')|} \leq \frac{m}{\ell} \implies |V(S')| \geq \frac{\ell|S'|}{m}.\]
    In particular, we can bound \eqref{eqn:Q_bound_for_S1} by
    \begin{align*}
        \sum_{\substack{S_2 \in \calS, \\ S_1 \cap S_2 \neq \varnothing}}q^{-|S_1\cap S_2|} 
        \leq \sum_{m' = 1}^m \sum_{\substack{S' \subseteq S_1, \\ |S'| = m'}}n^{\ell - |V(S')|}\,\ell^{r(m-m')}\,q^{-m'} 
        \leq n^\ell\,\ell^{rm}\sum_{m' = 1}^m \left(\frac{m}{\ell^r\,n^{\ell/m}\,q}\right)^{m'},
    \end{align*}
    where we use the fact that there are at most $m^{m'}$ choices for $S'$, given $S_1$.
    Note that the sum above is dominated by either the first or last term. Therefore, we get
    \[ \sum_{\substack{S_2 \in \calS, \\ S_1 \cap S_2 \neq \varnothing}} q^{-|S_1\cap S_2|}  \leq mn^\ell\,\ell^{rm}\max\left\{\frac{m}{\ell^r\,n^{\ell/m}\,q}, \left(\frac{m}{\ell^r\,n^{\ell/m}\,q}\right)^m\right\}.\]
    With this in hand, we return to our original computation of the variance to get
    \begin{align*}
        \Var_\calQ(T) &\leq m\,N\,q^{2m}\,n^\ell\,\ell^{rm}\max\left\{\frac{m}{\ell^r\,n^{\ell/m}\,q}, \left(\frac{m}{\ell^r\,n^{\ell/m}\,q}\right)^m\right\} \\
        &= \max\left\{m^2\,N\,n^{\ell(1 - 1/m)}\,\ell^{r(m-1)}\,q^{2m - 1}, m^{m+1}N\,q^{m}\right\},
    \end{align*}
    as desired.
\end{proof}

Before studying the expectation and variance of $T$ under $\calP$, we establish the following lemma.

\begin{lemma}
\label{lem:psi-s-1-s-2-exp-p}
Fix $S_1, S_2 \in \calS$ as defined in \eqref{eq:def-t-statistic-2}.
Assume \eqref{eqn:small_gamma_ub} and \eqref{eq:edge-vertex-ratio-between}.
Then we have
$$
\EE_\calP\left[\psi_{S_1\cup S_2}(Y)\right] \le 2^{2 \ell} \, \rho^{|V(S_1\cup S_2)|} \, p^{|S_1\cup S_2|} .
$$
\end{lemma}

\begin{proof}
By considering all possible realizations of $V(S_1 \cup S_2) \cap Z$, we have 
\begin{align*}
\EE_\calP\left[\psi_{S_1\cup S_2}(Y)\right] 
&= \sum_{V'\subseteq V(S_1\cup S_2)} \calP( V(S_1 \cup S_2) \cap Z = V' ) \, \calP( \psi_{S_1 \cup S_2}(Y) = 1 \mid V(S_1 \cup S_2) \cap Z = V' ) \\ 
&\leq \sum_{V'\subseteq V(S_1\cup S_2)} \rho^{|V'|}p^{|E(V')|}\,q^{|S_1\cup S_2| - |E(V')|} ,
\end{align*}
where $E(V')$ consists of the edges in $S_1\cup S_2$ whose vertices are contained entirely in $V'$.
We claim that the summand above is maximized when $V' = V(S_1\cup S_2)$.
Given the claim, each summand is bounded by $\rho^{|V(S_1 \cup S_2)|} \, p^{|S_1 \cup S_2|}$, and there are at most $2^{2 \ell}$ choices of $V'$.
Thus, the result follows.

To prove the claim, we first note that if $V'$ maximizes $\rho^{|V'|}p^{|E(V')|}\,q^{|S_1\cup S_2| - |E(V')|}$, then the subgraph $(V', E(V'))$ has no isolated vertices.
This is because removing isolated vertices from $V'$ only decreases the exponent in $\rho$, which then increases the overall value as $\rho < 1$.
As a result, we may instead consider the problem
$$
\max_{S' \subseteq S_1 \cup S_2} \rho^{|V(S')|} \, p^{|S'|} \, q^{|S_1\cup S_2| - |S'|} ,
$$
where $V(S')$ is the set of vertices induced by $S'$.
We will show that the maximizer is $S' = S_1 \cup S_2$.
We make the following definitions:
\[S \defeq S_1 \cup S_2 , \quad S_1' \defeq S'\cap S_1, \quad S_2' \defeq S' \setminus S_1'.\]
In particular, $S_i' \subseteq S_i$ and $S_2' \cap S_1' = S_2' \cap S_1 = \varnothing$.
The goal is to show the following:
\[\rho^{|V(S)|}p^{|S|} \geq \rho^{|V(S')|}p^{|S'|}\,q^{|S| - |S'|},\]
which is equivalent to showing
\begin{align}\label{eqn:max_S_bound}
\rho^{|V(S)| - |V(S')|}\left(\frac{p}{q}\right)^{|S| - |S'|} \geq 1.
\end{align}

Let us first consider the exponent of $p/q$. We have
\[|S| - |S'| = |S_1| + |S_2| - |S_1\cap S_2| - |S_1'| - |S_2'| = |S_1| - |S_1'| + |S_2| - |(S_1\cap S_2) \cup S_2'|,\]
as $S_2'\cap S_1\cap S_2 = \varnothing$.
Next, let us consider the exponent of $\rho$.
We have
\begin{align*}
|V(S)| - |V(S')| &= |V(S_1)| + |V(S_2)| - |V(S_1) \cap V(S_2)| - |V(S_1')| - |V(S_2')| + |V(S_1') \cap V(S_2')| \\
&= |V(S_1)| - |V(S_1')| + |V(S_2)| - |V(S_1) \cap V(S_2)| - |V(S_2')| + |V(S_1') \cap V(S_2')|.
\end{align*}
We note the following:
\[V(S_1') \cap V(S_2') \subseteq V(S_1) \cap V(S_2) \cap V(S_2'),\]
as $S_i' \subseteq S_i$.
Therefore,
\begin{align*}
&|V(S)| - |V(S')| \\
&\le |V(S_1)| - |V(S_1')| + |V(S_2)| - |V(S_1) \cap V(S_2)| - |V(S_2')| + |V(S_1) \cap V(S_2) \cap V(S_2')| \\
&\leq |V(S_1)| - |V(S_1')| + |V(S_2)| - |(V(S_1) \cap V(S_2))\cup V(S_2')| \\
&\leq |V(S_1)| - |V(S_1')| + |V(S_2)| - |V((S_1\cap S_2)\cup S_2')|,
\end{align*}
where the last inequality follows from the fact that
\[V((S_1\cap S_2) \cup S_2') 
= V(S_1\cap S_2) \cup V(S_2') 
\subseteq (V(S_1) \cap V(S_2))\cup V(S_2').\]
We can now simplify the LHS of \eqref{eqn:max_S_bound} to get:
\begin{align*}
&~\rho^{|V(S)| - |V(S')|}\left(\frac{p}{q}\right)^{|S| - |S'|} \\
&\geq \rho^{|V(S_1)| - |V(S_1')| + |V(S_2)| - |V((S_1\cap S_2)\cup S_2')|}\left(\frac{p}{q}\right)^{|S_1| - |S_1'| + |S_2| - |(S_1\cap S_2) \cup S_2'|} \\
&= \rho^{|V(S_1)| - |V(S_1')|}\left(\frac{p}{q}\right)^{|S_1| - |S_1'|}\rho^{|V(S_2)| - |V((S_1\cap S_2)\cup S_2')|}\left(\frac{p}{q}\right)^{|S_2| - |(S_1\cap S_2) \cup S_2'|}. 
\end{align*}
Hence, it suffices to show that
\begin{subequations}
\begin{align}
&\rho^{|V(S_1)| - |V(S_1')|} \left(\frac{p}{q}\right)^{|S_1| - |S_1'|} \ge 1 , \label{eq:two-terms-at-least-1} \\
&\rho^{|V(S_2)| - |V((S_1\cap S_2)\cup S_2')|} \left(\frac{p}{q}\right)^{|S_2| - |(S_1\cap S_2) \cup S_2'|} \ge 1 . \label{eq:two-terms-at-least-2}
\end{align}
\end{subequations}

We now prove \eqref{eq:two-terms-at-least-1}.
The proof of \eqref{eq:two-terms-at-least-2} is the same once $S_1$ is replaced by $S_2$ and $S_1'$ is replaced by $(S_1\cap S_2)\cup S_2'$.
First, if $V(S_1') = V(S_1)$, then the LHS of \eqref{eq:two-terms-at-least-1} can be simplified to $(\frac{p}{q})^{|S_1| - |S_1'|}$, which is at least $1$ as $p > q$.
Therefore, we may assume that $V(S_1') \subsetneq V(S_1)$.
Recall that $K_n^r[S_1]$ is isomorphic to $\tilde H$, a balanced hypergraph with $\ell$ vertices and $m$ edges.
By Claim \ref{claim:balanced_hypergraph_complement}, we have
\begin{align*}
\rho^{|V(S_1)| - |V(S_1')|} \left(\frac{p}{q}\right)^{|S_1| - |S_1'|}
= \left(\rho\left(\frac{p}{q}\right)^{\frac{|S_1| - |S_1'|}{|V(S_1)| - |V(S_1')|}}\right)^{|V(S_1)| - |V(S_1')|} 
\geq \left(\rho\left(\frac{p}{q}\right)^{m/\ell}\right)^{|V(S_1)| - |V(S_1')|}.
\end{align*}
Note that
\[\rho\left(\frac{p}{q}\right)^{m/\ell} = n^{\gamma - 1 + (\beta - \alpha)m/\ell}.\]
It is enough to show the exponent above is nonnegative, which holds if $(\beta - \alpha)m \geq (1-\gamma)\ell$.
From \eqref{eqn:small_gamma_ub} and \eqref{eq:edge-vertex-ratio-between}, we have
\[(\beta - \alpha)m > (\beta - \beta\gamma)m = (1- \gamma)\beta m > (1-\gamma)\ell,\]
as desired.
\end{proof}

Let us now bound expectation and variance of $T$ under $\calP$.

\begin{lemma}\label{lemma:P_bounds_T}
    With $T$ as defined in \eqref{eq:def-t-statistic-2}, we have
    \[\EE_\calP[T] \geq N\,\rho^\ell p^m, \quad \Var_\calP(T) \leq 8^{\ell}\,N\,n^{\ell - 1}\,\ell^{1 + rm}\,\rho^{2\ell - 1}p^{2m - m/\ell},\]
    where $N = |\calS|$.
\end{lemma}

\begin{proof}
    As $Y_e$ and $Y_f$ are conditionally independent for $e \neq f$ given $Z$, we note that for a specific $S \in \calS$, we have
    \[\calP\left(\psi_{S}(Y) = 1 \mid V(S) \subseteq Z\right) = p^m.\]
    In addition, $\psi_S(Y) \in \{0,1\}$, so we have
    \begin{equation*}
        \EE_\calP[T] \geq \sum_{S \in \calS} \calP(V(S) \subseteq Z) \, \calP\left(\psi_{S}(Y) = 1 \mid V(S) \subseteq Z\right) 
        = N\rho^\ell p^m.
    \end{equation*}
    
    For the variance, we note that $z_i$ and $z_j$ are independent for $i \neq j$ and so $\psi_{S_1}(Y)$ and $\psi_{S_2}(Y)$ are independent if $V(S_1) \cap V(S_2) = \varnothing$.
    Using again the argument that gives \eqref{eq:var-s-1-s-2-union}, we obtain
    \begin{align*}
        \Var_\calP(T) = \sum_{\substack{S_1, S_2 \in \calS, \\ V(S_1) \cap V(S_2) \neq \varnothing}}\Cov_\calP\left(\psi_{S_1}(Y), \psi_{S_2}(Y)\right) 
        \leq \sum_{\substack{S_1, S_2 \in \calS, \\ V(S_1) \cap V(S_2) \neq \varnothing}}\EE_\calP\left[\psi_{S_1\cup S_2}(Y)\right] .
    \end{align*}
By Lemma~\ref{lem:psi-s-1-s-2-exp-p}, we then get 
    \begin{align}
        \Var_\calP(T) &\leq \sum_{\substack{S_1, S_2 \in \calS, \\ V(S_1) \cap V(S_2) \neq \varnothing}}2^{2\ell} \rho^{|V(S_1\cup S_2)|}p^{|S_1\cup S_2|} \notag \\
        &= 2^{2\ell}\,\rho^{2\ell}\,p^{2m}\,\sum_{\substack{S_1, S_2 \in \calS, \\ V(S_1) \cap V(S_2) \neq \varnothing}}\rho^{-|V(S_1) \cap V(S_2)|}p^{-|S_1\cap S_2|}, \label{eq:var-p-t-bound-intermediate}
    \end{align}
    where we use the facts 
    \[|S_1 \cup S_2| = 2m - |S_1 \cap S_2|, \quad |V(S_1 \cup S_2)| = |V(S_1) \cup V(S_2)| = 2\ell - |V(S_1) \cap V(S_2)|.\]
    As in the proof of Lemma \ref{lemma:Q_bounds_T}, let us fix $S_1$ and bound the following:
    \begin{align}\label{eqn:P_bound_for_S1}
        \sum_{\substack{S_2 \in \calS, \\ V(S_1) \cap V(S_2) \neq \varnothing}}\rho^{-|V(S_1) \cap V(S_2)|}p^{-|S_1\cap S_2|} \leq \sum_{\ell' = 1}^\ell \sum_{\substack{V' \subseteq V(S_1), \\ |V'| = \ell'}}\sum_{\substack{S_2 \in \calS, \\ V' \subseteq V(S_2)}}\rho^{-\ell'} p^{-|S_1\cap S_2|}.
    \end{align}
    Let us first consider a fixed $V'$ and bound the number of choices for $S_2$ in the inner-most sum.
    There are at most $n^{\ell - \ell'}$ choices for the remaining vertices in $S_2$ and at most $\ell^{rm}$ possible choices for the edges in $S_2$.
    Note that $V(S_1 \cap S_2) \subseteq V(S_1) \cap V(S_2)$.
    Furthermore, as $S_1\cap S_2 \subseteq S_2$ and $K_n^r[S_2]$ is a balanced hypergraph, we have
    \[\frac{|S_1\cap S_2|}{|V(S_1) \cap V(S_2)|} \leq \frac{|S_1\cap S_2|}{|V(S_1 \cap S_2)|} \leq \frac{m}{\ell}.\]
    In particular, we can bound \eqref{eqn:P_bound_for_S1} by
    \begin{align*}
        \sum_{\substack{S_2 \in \calS, \\ V(S_1) \cap V(S_2) \neq \varnothing}}\rho^{-|V(S_1) \cap V(S_2)|}p^{-|S_1\cap S_2|} 
        &\leq \sum_{\ell' = 1}^\ell \sum_{\substack{V' \subseteq V(S_1), \\ |V'| = \ell'}}n^{\ell - \ell'}\,\ell^{rm}\,\rho^{-\ell'}p^{-m\ell'/\ell} \\
        &\leq n^\ell\,\ell^{rm}\sum_{\ell' = 1}^\ell\left(\frac{\ell}{n\rho\,p^{m/\ell}}\right)^{\ell'},
    \end{align*}
    where we use the fact that there are at most $\ell^{\ell'}$ choices for $V'$ given $S_1$. 
    Let us consider the term in the sum.
    We have
    \[\frac{\ell}{n\rho\,p^{m/\ell}} = \ell\,n^{\alpha\,m/\ell - \gamma}.\]
    Note that as a result of \eqref{eq:edge-vertex-ratio-between}, we have
    \[\alpha\frac{m}{\ell} - \gamma < 0.\]
    In particular, we have for some small constant $\delta = \delta(\alpha, \beta, \gamma) > 0$:
    \[\frac{\ell}{n\rho\,p^{m/\ell}} \leq \ell\,n^{-\delta} < \frac{1}{2}.\]
    Finally, we can bound \eqref{eqn:P_bound_for_S1} as follows:
    \[\sum_{\substack{S_2 \in \calS, \\ V(S_1) \cap V(S_2) \neq \varnothing}}\rho^{-|V(S_1) \cap V(S_2)|}p^{-|S_1\cap S_2|} \leq 2\,n^{\ell}\,\ell^{rm}\left(\frac{\ell}{n\rho\,p^{m/\ell}}\right).\]
    With this in hand, we return to our original computation for the variance in \eqref{eq:var-p-t-bound-intermediate} to get
    \[\Var_\calP(T') \leq 2^{2\ell+1}\,N\,n^\ell\,\ell^{rm}\rho^{2\ell}\,p^{2m}\left(\frac{\ell}{n\rho\,p^{m/\ell}}\right) \leq 8^\ell\,N\,n^{\ell - 1}\,\ell^{1 + rm}\,\rho^{2\ell - 1}p^{2m - m/\ell},\]
    as desired.
\end{proof}

\begin{proof}[Proof of Proposition~\ref{prop:upper-small-gamma}]
Let $\lambda \defeq N\,\rho^\ell\,p^m$.
By Lemma~\ref{lemma:P_bounds_T}, we have $\EE_\calQ[T] \ge \lambda$.
Therefore, it suffices to show that
\begin{enumerate}
\item\label{item:E_Q} $\EE_\calQ[T] = o(\lambda)$, 
\item\label{item:Var_Q} $\Var_\calQ(T) = o(\lambda^2)$, 
\item\label{item:Var_P} $\Var_\calP(T) = o(\lambda^2)$.
\end{enumerate}

Let us first consider $\EE_\calQ[T]$. 
By Lemma~\ref{lemma:Q_bounds_T}, we have
\[\frac{\EE_\calQ[T]}{\lambda} = \frac{Nq^m}{N\rho^\ell p^m} = n^{(1 - \gamma)\ell -(\beta - \alpha)m}
< n^{(1 - \gamma) \ell - \beta(1 - \gamma) m} = n^{(1 - \gamma) (\ell - \beta m)} = o(1), \]
where the bounds follow from the conditions $\alpha < \beta \gamma$ in \eqref{eqn:small_gamma_ub}, $\ell < m \beta$ in \eqref{eq:edge-vertex-ratio-between}, and $1-\gamma > 0$.
This completes the proof of Condition~\ref{item:E_Q}.
    
Now, let us consider the bound on $\Var_\calQ(T)$ from Lemma~\ref{lemma:Q_bounds_T}. 
We note the following
\[N \geq \binom{n}{\ell} \geq \left(\frac{n}{\ell}\right)^\ell.\]
With this and the above bound on $\EE_\calQ[T]$, we have 
\[\frac{m^{m+1}\EE_\calQ[T]}{\lambda^2} \ll \frac{m^{m+1}}{\lambda} \leq \frac{m^{m+1}\ell^\ell}{n^\ell\,\rho^\ell\,p^m} = n^{\alpha m - \gamma \ell + o(1)} = o(1) ,\]
because $\alpha m < \gamma \ell$ by \eqref{eq:edge-vertex-ratio-between}. 
Now, in order to prove Condition~\ref{item:Var_Q}, it remains to bound the following:
\begin{align}
    \frac{m^2\,N\,n^{\ell(1 - 1/m)}\,\ell^{r(m-1)}\,q^{2m - 1}}{\lambda^2} &= \frac{m^2\,n^{\ell(1 - 1/m)}\,\ell^{r(m-1)}\,q^{2m - 1}}{N\rho^{2\ell}p^{2m}} \notag \\
    &\leq \frac{m^2\,n^{-\ell/m}\,\ell^{\ell + r(m-1)}\,q^{2m - 1}}{\rho^{2\ell}p^{2m}} \notag \\
    &= n^{-\ell/m + \beta + 2\ell(1-\gamma) - 2m(\beta - \alpha) + o(1)}. \label{eq:n-exponent-negative-several-terms}
\end{align}
Let us consider the exponent in \eqref{eq:n-exponent-negative-several-terms}.
As a result of \eqref{eqn:small_gamma_ub}, we have
\begin{align*}
    -\ell/m + \beta + 2\ell(1-\gamma) - 2m(\beta - \alpha) &= \left(1-\gamma - \frac{1}{2m}\right)2\ell - \left(\beta\left(1 - \frac{1}{2m}\right) - \alpha\right)2m \\
    &< \left(1-\gamma - \frac{1}{2m}\right)2\ell - 2m \beta\left(1 - \gamma - \frac{1}{2m}\right) \\
    &= 2\left(1-\gamma - \frac{1}{2m}\right)(\ell - m\beta).
\end{align*}
Since $\gamma < 1/2$ and $m \geq 1$, we have $1 - \gamma - 1/(2m) > 0$.
In addition, $\ell < m \beta$ as a result of \eqref{eq:edge-vertex-ratio-between}. 
Hence, the quantity above is negative, and so \eqref{eq:n-exponent-negative-several-terms} is $o(1)$.
Putting it all together, we have proved $\frac{\Var_{\calQ}(T)}{\lambda^2} = o(1)$, which is Condition~\ref{item:Var_Q}.

Finally, let us consider the bound on $\Var_\calP(T)$ in Lemma~\ref{lemma:P_bounds_T}.
We have
\begin{align*}
    \frac{\Var_\calP(T)}{\lambda^2} &\leq \frac{8^\ell\,N\,n^{\ell - 1}\,\ell^{1 + rm}\,\rho^{2\ell - 1}p^{2m - m/\ell}}{N^2\rho^{2\ell}p^{2m}} 
    \leq \frac{8^\ell\,\ell^{1 + rm + \ell}}{n\rho\,p^{m/\ell}} 
    = n^{\alpha m/\ell - \gamma + o(1)} 
    = o(1) ,
\end{align*}
where we again use $N \geq \left(\frac{n}{\ell}\right)^\ell$ and $\alpha m / \ell - \gamma < 0$ by \eqref{eq:edge-vertex-ratio-between}.
This proves Condition~\ref{item:Var_P}.
\end{proof}

\appendix

\section{Detection-refutation gap}
\label{app:det-ref}

Here we justify the claim that for $\gamma > 1/2$ we expect an inherent detection-refutation gap, implying that tests based on checking feasibility of a convex relaxation (e.g., the sum-of-squares program of~\cite{sos-dense-subgraph}) are strictly suboptimal for detection. We focus throughout on the graph case $r=2$.

We first formally define the refutation task. We say that a graph contains a ``dense subgraph'' if there is a subgraph on $(1 \pm o(1)) \rho n$ vertices with $(1 \pm o(1)) p (\rho n)^2/2$ edges, where each $o(1)$ stands for some particular $o(1)$ quantity chosen so that the planted distribution $\calP$ (defined in Section~\ref{section:main_results}) will contain a dense subgraph with high probability (i.e., probability $1-o(1)$). Refutation is the following algorithmic task. Given a graph $G$, the goal is to output NO or MAYBE with the following two guarantees: (1) if $G$ contains a dense subgraph then the output must be MAYBE, and (2) if $G$ is drawn from the null distribution $\calQ = G(n,q)$, the output must be NO with probability $1-o(1)$. For intuition, it is crucial to note that the algorithm is only allowed to output NO if it has proven with absolute certainty that there is no dense subgraph. Refutation (or certification) tasks of this flavor have been considered in prior work, e.g., \cite{rip-cert,sk-cert,local-stats,spectral-planting,nonneg-pca,coloring-clique}.

A natural approach to solve the refutation problem is to check feasibility of a convex relaxation for existence of a dense subgraph. If the relaxation is infeasible w.h.p.\ over $\calQ$ (for some parameters $p,q,\rho$), we have a successful refutation algorithm: output NO if the relaxation is infeasible, and MAYBE otherwise. (Note that by virtue of being a relaxation, if the relaxation is infeasible then this proves there is no dense subgraph.)  This also gives a successful detection algorithm: output ``$\calQ$'' if the relaxation is infeasible, and ``$\calP$'' otherwise. (Note that by design, the relaxation will be feasible w.h.p.\ over $\calP$, because a dense subgraph exists.)

Our goal in this section is to argue that for $\gamma > 1/2$, the refutation problem is computationally hard whenever $\alpha > \beta/2 + \gamma - 1/2$, making it strictly harder than detection (because detection is easy whenever $\alpha < \beta/2 + 2 (\gamma - 1/2)$). As a result, we expect that \emph{any} test based on checking feasibility of a convex relaxation (in the sense described above) cannot be optimal for detection, or else this would imply a too-good-to-be-true refutation algorithm. Indeed, if the relaxation succeeds at detection, it must be infeasible w.h.p.\ over $\calQ$, meaning it also succeeds at refutation.

We will argue hardness of refutation in a manner similar to~\cite{sk-cert}: we will construct a \emph{different} planted distribution $\widetilde\calP$ that, like $\calP$, contains a dense subgraph w.h.p.; we will then show low-degree hardness of distinguishing $\widetilde\calP$ from $\calQ$, leading us to conjecture that no polynomial-time algorithm can distinguish $\widetilde\calP$ and $\calQ$. This conjecture, if true, formally implies hardness of refutation because a successful refutation algorithm could be used to distinguish $\widetilde\calP$ and $\calQ$ as discussed above.

To summarize, our goal for the rest of this section is to construct a distribution $\widetilde\calP$ over graphs such that a dense subgraph exists w.h.p., and then prove that if $\gamma > 1/2$ and $\alpha > \beta/2 + \gamma - 1/2$ then no degree-$n^{o(1)}$ polynomial weakly separates $\widetilde\calP$ and $\calQ$.

\paragraph{Construction of $\widetilde\calP$.}

We now construct an auxiliary planted distribution $\widetilde\calP$ that has a planted dense subgraph but is more difficult to distinguish from $\calQ = G(n,q)$ than $\calP$ is. As in the main text, we fix parameters $0 < \alpha < \beta < 1$ and $\gamma \in (0,1)$, and consider the scaling $p = n^{-\alpha}$, $q = n^{-\beta}$, and $\rho = n^{\gamma-1}$. It will be convenient to parametrize the observed graph in an unusual way: for $i<j$, let $Y_{ij} = a \defeq \sqrt{(1-q)/q}$ if edge $(i,j)$ is present and $Y_{ij} = b \defeq -\sqrt{q/(1-q)}$ otherwise. (This ensures $\EE_\calQ[Y_{ij}] = 0$ and $\EE_\calQ[Y_{ij}^2] = 1$.) A graph is sampled from $\widetilde\calP$ as follows. As before, sample a set of planted vertices $Z \subseteq [n]$ where each vertex is included independently with probability $\rho$. Define $u \in \R^n$ by $u_i = \sqrt{(1-\rho)/\rho}$ if $i \in Z$ and $u_i = -\sqrt{\rho/(1-\rho)}$ if $i \notin Z$. (This ensures $\EE[u_i] = 0$ and $\EE[u_i^2] = 1$.) Now, conditioned on $u$, let $Y = (Y_{ij})_{i < j}$ have independent entries $Y_{ij} \in \{a,b\}$ such that $\EE_{\tilde \calP}[Y_{ij} | u] = \lambda u_i u_j$. Choose the scalar $\lambda$ so that planted edges have the desired probability, i.e., edge $(i,j)$ occurs with probability $p$ when $i,j \in Z$; this gives $\lambda = (1+o(1)) p \rho / \sqrt{q}$.

\paragraph{Proof of low-degree hardness.}

We now show that if $\gamma > 1/2$, $\alpha > \beta/2 + \gamma - 1/2$, and $D = n^{o(1)}$, then $\|\widetilde L_{\le D}\| = 1+o(1)$, which rules out weak separation as discussed in Section~\ref{section:lower_bound}. By~\cite[Proposition~B.1]{spectral-planting} along with the fact that revealing extra observations $(Y_{ij})_{i \ge j}$ can only increase $\|\widetilde L_{\le D}\|$,
\[ \|\widetilde L_{\le D}\|^2 \le \sum_{d=0}^D \frac{1}{d!} \EE \langle \lambda uu^\top,\lambda vv^\top \rangle^d = \sum_{d=0}^D \frac{\lambda^{2d}}{d!} \EE \langle u,v \rangle^{2d}, \]
where $u$ is defined above and $v$ is an independent copy of $u$. Note that $\langle u,v \rangle$ is a sum of $n$ i.i.d.\ bounded random variables and following the proof of~\cite[Lemma~B.3]{SW-estimation}, we may conclude
\[ \EE\langle u,v \rangle^{2d} \le \sqrt{2\pi} \rho^{-2d} \left[(4d \rho^2 n)^d + (8d/3)^{2d}\right]. \]
Combining this with the above, $\|\widetilde L_{\le D}\| = 1+o(1)$ provided $\lambda \ll \rho$ and $\lambda^2 n \ll 1$. Since $\gamma > 1/2$, this reduces to the condition $\alpha > \beta/2 + \gamma - 1/2$ as desired.

\subsection*{Acknowledgments}

We thank Guy Bresler and Aaron Potechin for helpful discussions.

\bibliographystyle{alpha}
\bibliography{main}

\end{document}